\documentclass{article}
\usepackage{authblk}
\usepackage[english]{babel}
\usepackage[letterpaper,top=2cm,bottom=3cm,left=3cm,right=3cm,marginparwidth=1.75cm]{geometry}
 
\usepackage{amssymb,amsthm}
\usepackage{bbm}
\usepackage{amsmath}
\usepackage{todonotes}
\usepackage{mathrsfs}
\usepackage{thmtools}
 
\numberwithin{equation}{section}
\usepackage{graphicx}
\usepackage[colorlinks=true, allcolors=blue]{hyperref}
 
\providecommand{\argmin}{\mathop\mathrm{arg min}}
 
\usepackage{algorithm}
\usepackage{algpseudocode}
 
\newtheorem{theorem}{Theorem}[section]
\newtheorem{lemma}{Lemma}[section]
\newtheorem{corollary}{Corollary}[section]
\newtheorem{proposition}{Proposition}[section]
\newtheorem{definition}{Definition} [section]
\newtheorem{assumption}{Assumption}[section]
\newtheorem{remark}[theorem]{Remark}
\newtheorem{example}{Example}[section]
\def\rmd{\mathrm{d}}
\def\Law{\mathrm{Law}}

\DeclareMathOperator{\R}{\mathbb{R}}

\newcommand{\E}{\mathbb{E}}

\def\K{\mathcal{K}}
\def\uu{f}
 
\def\btheta{\boldsymbol\theta}
\newcommand{\proj}{\mathtt{P}}
\newcommand{\wsq}{{\sf W}_q}
\newcommand{\wsone}{{\sf W}_1}
\newcommand{\wstwo}{{\sf W}_2}
\newcommand{\wsharp}{{\sf W}_{\sharp}}
\newcommand{\bigo}{\mathcal{O}}
\newcommand{\logo}{\widetilde{\mathcal{O}}}

\def\bvartheta{\boldsymbol\vartheta}
\def\bv{\mathbf v}
\def\bxi{\boldsymbol\xi}
\def\bfC{\mathbf C}
\def\bfA{\mathbf A}
\def\bfB{\mathbf B}
\def\bV{\mathbf V}
\def\bL{\boldsymbol L}
\def\bW{\boldsymbol W}

\def\bz{\mathbf z}
\def\Ltwo{\mathbb L_2}
\newcommand{\nsharp}[1]{\|#1\|_{\sharp}}

\newcommand{\etah}{\varrho_h}
\title{Kinetic Langevin Splitting Schemes for  Constrained Sampling}
\author{Neil K. Chada\thanks{Department of Mathematics, City University of Hong Kong, Hong Kong SAR, \texttt{neilchada123@gmail.com}} \qquad Lu Yu\thanks{Data Science Institute, Shandong University, P.R. China \texttt{lu.yu@sdu.edu.cn}}}
\date{\today}
\begin{document}
\maketitle
\begin{abstract}
Constrained sampling, the task of generating samples from a distribution restricted to a convex set, is an important problem in computational statistics. In this paper, we propose new unadjusted sampling algorithms obtained by applying the symmetric BAOAB and ABOBA splittings, of the kinetic Langevin dynamics, to a penalized surrogate of the constrained target. This also includes their stochastic gradient versions. Since the Hessian of the penalized surrogate is discontinuous across the boundary of the constraint set, analyses relying on a Lipschitz Hessian do not apply. We therefore develop an analysis based only on strong convexity and smoothness, proving one-step Wasserstein contraction schemes, non-asymptotic Wasserstein bounds of both schemes and complexity bounds.  The resulting complexities for the stochastic gradient versions improve substantially on traditional Euler-type discretizations. Numerical experiments verify our theory, which is tested on a toy problem and Bayesian constrained regression problem. Overall, the splitting schemes demonstrate stability at step sizes several times larger than Euler-type methods, require five to seventeen times fewer gradient evaluations for a given accuracy.
\\\\
\textbf{Keywords}: Constrained sampling, Splitting schemes, Unadjusted Langevin algorithms, \\ Wasserstein complexity
\end{abstract}
 
\tableofcontents
 
\section{Introduction}
 
Sampling from probability distributions plays a critical role in various fields of science and engineering. Examples of such fields include numerical weather prediction, geophysical sciences, and more recently machine learning, such as generative modelling or Bayesian neural networks \cite{score,wideBNN,dick}. Most of these applications exploit Monte Carlo methodologies for sampling, with limited, or rarely any significant, restrictions. However, in many situations it is of interest to consider sampling with dealing with convex or compact sets, in other words, restricting the sampling to a particular region within the sample space. We will refer to this task as \textit{constrained sampling}, which will be the focus of this work.
Many scenarios involve sampling on constrained, or confined, spaces, for example stochastic optimal control and molecular dynamics \cite{regression2012,leimkuhler2015molecular}.
In this context, the problem involves sampling from a probability measure (which we often refer to as the target) $\nu$ on such sets, which are characterized by a density function of the form
\begin{equation}
\label{eq:density}
\nu(\btheta)=\frac{e^{-U^{}(\btheta)}}{\int_{\mathbb{R}^p} e^{-U^{}(\btheta)}\rmd \btheta}\,,
\end{equation}
where from the equation, $U^{}(\btheta)$ takes the form $U(\btheta) = f(\btheta) + \ell_{\mathcal{K}}(\btheta)$, where $f(\btheta)$ represents a potential function and $\ell_{\mathcal{K}}(\btheta)$ is an indicator function ensuring the parameters $\btheta$ lies within a convex and compact set $\mathcal{K}\subset \mathbb{R}^p$, for $p \geq 1$,
\begin{equation}
\label{eq:indicator}
\ell_{\K}(\btheta):=
\begin{cases}
+\infty & \text{ if } \btheta\notin\K\\
 0& \text{ if } \btheta\in\K\,.
\end{cases}
\end{equation}
 
The absence of smoothness in the target distribution $\nu$ presents significant challenges because sampling algorithms often depend, quite heavily, on the smoothness properties of the target distribution. This is to ensure that one can effectively explore the state space and produce representative samples of the target $\nu$. Motivated by this, we consider the approximation for $\ell_{\K}$ of the form
\begin{align*}
    \ell_{\K}^{\lambda}(\btheta):=\dfrac{1}{2\lambda^2}d_{\K}(\btheta),
\end{align*}
where $\lambda>0$ is the tuning parameter and
$d_{\K}:\mathbb{R}^p\to\mathbb R_+$ is a distance function that quantifies the distance between $\btheta$ and the constraint set $\K$.
We will later introduce specific choices for $d_{\K}.$
Now we define
\begin{equation}
\label{eq:new_pot}
U^\lambda(\btheta):=f(\btheta)+\ell_{\K}^\lambda(\btheta)\,,
\end{equation}
and the corresponding surrogate target density $\nu^\lambda$ given as
\begin{equation}
\label{eq:surrogate}
\nu^\lambda(\btheta):=\frac{e^{-U^{\lambda}(\btheta)}}{\int_{\mathbb{R}^p} e^{-U^{\lambda}(\btheta')}\rmd \btheta'}\,.
\end{equation}
Therefore, we consider sampling from the target density~$\nu^\lambda\propto \exp(-U^\lambda)$. A popular approach for this task are kinetic Langevin methods, or referred to as unadjusted Langevin algorithms (ULA). These methods are based on the kinetic Langevin dynamics (KLD) (also referred to as underdamped Langevin dynamics \cite{cheng2018underdamped,dalalyan2020sampling}), a well-known approach to model the dynamics of molecular systems. It has been heavily utilized in different mathematical areas, which include molecular dynamics, Bayesian statistics, and more recently machine learning. It is defined through two coupled processes $\{\bL_t\}_{t \geq 0}$ and $\{\bV_t\}_{t \geq 0}$ which are the position and the velocity, defined as the following
\begin{equation}\label{eq:kinetic_langevin}
    \begin{split}
    d\bL_{t} &= \bV_{t}dt,\\
    d\bV_{t} &= -\nabla U(\bL_{t}) dt - \gamma \bV_{t}dt + \sqrt{2\gamma}d\bW_{t},
\end{split}
\end{equation}
where $U:\R^p\to\R$ is a potential energy function, $\{\bW_t\}_{t\geq 0}$ is a standard $p-$dimensional Brownian motion, and $\gamma>0$ is a friction coefficient. Under fairly weak assumptions, the unique invariant measure of the process $\{\bL_t,\bV_t\}_{t\geq 0}$ is of the form
\begin{equation}
\label{eq:inv}
\nu(d\bvartheta d\bv) \propto \exp\left(-U(\bvartheta)-\frac{ \|\bv\|^2}{2}\right) d\bvartheta d\bv.
\end{equation}
An alternative to KLD is to use the overdamped Langevin dynamics (OLD), which instead just considers the process $\{\bL_t\}_{t \geq 0}$ given as
\begin{equation}
\label{eq:OLD}
d\bL_t = - \nabla U(\bL_t)dt  + \sqrt{2\beta^{-1}} d\bW_t,
\end{equation}
where $\beta^{-1}>0$ is the temperature parameter. The invariant measure associated with \eqref{eq:OLD} is
$
\nu(d\bvartheta) \propto \exp\left(-\beta U(\bvartheta)\right) d\bvartheta,
$
which coincides with the target for $\beta=1$.
The benefit of using the KLD over the OLD, is that it is known to converge faster to its respective invariant measure \cite{Ma2021,cheng2018underdamped}.
In the context of Bayesian sampling, both \eqref{eq:kinetic_langevin} and \eqref{eq:OLD} can be used to generate samples from a distribution $\nu$. This is achieved by simply setting the potential as the negative log-posterior, i.e. $U=-\log \nu$. Recent work \cite{Ma2021, Wibisono2018} has looked at connecting these processes to well-known optimization procedures. Implementing such dynamics
with the full continuous process is challenging, and often impossible. Therefore one must resort to discretization schemes.
The most obvious discretization method is the Euler-Maruyama (EM) scheme, however existing results suggest alternative methods, based on the bias, complexity and stability. One such method is the randomized midpoint method (RMM) \cite{rmp}. The RMM has also provided advantages in error
tolerance and condition number dependence, demonstrated through applications in machine learning \cite{LJ24,YuYu1,Kandasamy2024RMP}. In the context of constrained sampling, this was used in the work of Yu et al. \cite{YuYu2}, where they demonstrate sharper approximation analysis, and improved error bounds in terms of Wasserstein distance.
 
Despite the success of the RMM applied to constrained sampling, there still remain limitations, such as that it uses two gradient evaluations per time-step. This promotes the use of schemes which use a single gradient evaluation per step while retaining good stability properties. For this, we introduce and motivate the use of an alternative series of schemes referred to as splitting schemes.
 
\subsection{Family of Splitting Schemes}
 
A popular choice for solving ODEs and SDEs are splitting schemes, motivated by Gilbert Strang \cite{strang}, where one splits the dynamics into different components and solves them individually. This is possible due to the fact they can be integrated exactly. This also translates to SDEs which include the KLD. Two families are popular for the KLD: (i) the BU splitting family \cite{BUBthesis}, and (ii) the BAO splitting family. The former splits the SDE \eqref{eq:kinetic_langevin} as
\[
\begin{pmatrix}
d\bvartheta \\
d\bv
\end{pmatrix} = \underbrace{\begin{pmatrix}
0 \\
-\nabla U(\bvartheta)dt
\end{pmatrix}}_{\mathcal{B}} +\underbrace{\begin{pmatrix}
\bv dt \\
-\gamma \bv dt + \sqrt{2\gamma}d\bW_t
\end{pmatrix}}_{\mathcal{U}},
\]
and composing exact solutions of the two parts gives, for example, the UBU scheme analyzed in \cite{sanz2021wasserstein}, which has strong order two. As noted in the correction \cite{paulin2024correction}, the second-order analysis of UBU requires a strengthened third-order regularity assumption on the potential. We explain in Remark~\ref{rem:regularity} why such regularity is not available for the penalized potentials $U^\lambda$ considered here, which is why we focus on the BAO family in this work.
 
In the BAO family, the solution maps of the three parts are denoted by $\mathcal{B}$, $\mathcal{A}$, and $\mathcal{O}$, and are defined, for a step of length $t>0$, as
\begin{equation}
\label{eq:BAO}
\mathcal{B}_t(\bvartheta,\bv) = (\bvartheta,\bv-t\nabla U(\bvartheta)), \quad \mathcal{A}_t(\bvartheta,\bv) = (\bvartheta+t \bv,\bv), \quad \mathcal{O}_t(\bvartheta,\bv) = \big(\bvartheta,e^{-\gamma t} \bv + \sqrt{1-e^{-2\gamma t}}\,
\bxi\big),
\end{equation}
where $\bxi\sim\mathcal N(0,I_p)$.
Examples of such schemes include BAOAB, ABOBA and OBABO \cite{leimkuhler2013jcp, leimkuhler2015molecular}, which are symmetric (palindromic) compositions with second order weak accuracy for smooth potentials, and which require only one new gradient evaluation per step.
A contraction and convergence analysis was presented for such schemes in \cite{leimkuhler2023contractiona, leimkuhler2023contractionb}, see also \cite{monmarche2021high} for OBABO, including the extension to stochastic gradients.
 
Thus far, there has been no connection in understanding splitting schemes for KLD in the context of constrained sampling, which acts as our motivation, in the penalized setting. Before we discuss our contributions of this paper, we provide a brief overview of related work on constrained sampling using various Langevin-based algorithms.
 
\subsection{Other Related Works}
 
The notion of constrained sampling can have many different interpretations, depending on the setup and assumptions placed.
Our setup, defined briefly thus far through \eqref{eq:surrogate}, is based on a penalized constrained setting first introduced in the work of \cite{Mert24}. Their motivation arises from penalty methods from continuous optimization, which includes a penalty term for constraint violations. However, this is not the only setup one could consider for constrained settings. A natural, but difficult setting could be to consider manifolds directly, which require sophisticated mathematical methods to effectively sample from the target. Recent work that has considered this, in the context of Langevin dynamics, includes \cite{kook2022samp,noble2024unbiased}. Some of these works require technical couplings, to ensure one can obtain both Wasserstein convergence and contraction. Furthermore, much of these works are based on accept-reject Markov chain Monte Carlo algorithms like Hamiltonian Monte Carlo. Despite the elegance of such proposed work, we consider a more simplified setting which is easier to begin with, before moving to a potential manifold setting in future work. Other such works include \cite{Leimkuhler2024confined}, where the authors develop confined BAO integrators for Langevin dynamics. It is the first paper aimed at exploiting such splitting schemes in the context of a confined space. However, it differs in that the dynamics are designed to reflect off boundaries, and the gradient information is encoded into each step of the integrator.  More recent literature includes that of \cite{nonreverse,reza}, where these papers consider a different setup, motivated through a primal-dual setting, and which do not include a KLD setting, i.e. only one underlying process $\{\bL_t\}_{t \geq 0}$.
 
\subsection{Contributions}
 
We highlight our contributions of this work through the following points below:
\begin{itemize}
    \item We introduce new algorithms for the task of constrained sampling, for sampling from \eqref{eq:density}. Specifically we introduce ULA-type algorithms based on the BAOAB and ABOBA splittings of the KLD applied to the penalized potential $U^\lambda$, which we call CBAOAB and CABOBA.
    \item We observe that the Hessian of $U^\lambda$ is discontinuous across $\partial\K$ for every convex set $\K$ and the projections considered, so that Hessian-Lipschitz assumptions used in second-order analyses of splitting schemes cannot be imposed. Our analysis uses only $m$-strong convexity and $M^\lambda$-smoothness of $U^\lambda$.
    \item We prove one-step Wasserstein contraction of both schemes in an explicit shifted weighted Euclidean norm, by conjugating the schemes to chains whose contraction was established in \cite{leimkuhler2023contractiona}. We then prove first-order bounds on the asymptotic bias of both schemes by a direct comparison with a scheme that preserves the invariant measure exactly, with explicit and moderate constants.
    \item We provide non-asymptotic Wasserstein bounds and complexity analysis for CBAOAB and CABOBA and for their stochastic gradient versions SG-CBAOAB and SG-CABOBA. For the stochastic gradient versions we compare the stochastic and full gradient chains through a one-step perturbation bound and sum the perturbations using the contraction of the full gradient kernel. The resulting complexities improve substantially on the stochastic gradient Euler-type scheme SG-CKLMC, for which we also provide a new analysis. A summary is provided in Table~\ref{table:summary}.
    \item Numerical experiments are provided for a circular constraint, a box-constrained Gaussian in which Wasserstein errors can be computed exactly in any dimension, and a Bayesian linear regression problem. They show that the splitting schemes are stable for step sizes several times larger than CKLMC, require five to seventeen times fewer gradient evaluations to reach a given accuracy, and are insensitive to the batch size, even for $b=1$. They also indicate that the bias of the splitting schemes behaves as $\mathcal O(h^2)$ for the squared-distance penalty, beyond what our first-order theory guarantees.
\end{itemize}
 
\begin{table}
\begin{center}
\begin{tabular}{ |c|c|c|c|c| }
\hline
\textbf{Algorithm} & \textbf{Complexity} & \textbf{Reference} & \textbf{Metric} & \textbf{Projection}\\
\hline
CLMC (PULMC) & $\mathcal{O}(\varepsilon^{-7})$ & \cite{Mert24} & $\wstwo$ & Euclidean\\
CKLMC & $\mathcal{O}(\varepsilon^{-4})$ & \cite{YuYu2} & $\wstwo$ & Euclidean\\
CBAOAB & $\logo(\varepsilon^{-4})$ &  This work, Theorem~\ref{thm:cbaoab} & $\wsone$ & Bregman, Gauge\\
CBAOAB & $\logo(\varepsilon^{-(7p+2)/(p+2)})$ &  This work, Theorem~\ref{thm:cbaoab} & $\wstwo$ & Bregman, Gauge\\
CABOBA & $\logo(\varepsilon^{-4})$ &  This work, Theorem~\ref{thm:caboba} & $\wsone$ & Bregman, Gauge\\
CABOBA & $\logo(\varepsilon^{-(7p+2)/(p+2)})$ &  This work, Theorem~\ref{thm:caboba} & $\wstwo$ & Bregman, Gauge\\
SG-CBAOAB & $\logo(\varepsilon^{\textcolor{red}{-6}})$ &  This work, Theorem~\ref{thm:CBAOAB_SG}  & $\wsone$ & Bregman, Gauge \\
SG-CBAOAB & $\logo(\varepsilon^{\textcolor{red}{-(9p+6)/(p+2)}})$ &  This work, Theorem~\ref{thm:CBAOAB_SG}  & $\wstwo$ & Bregman, Gauge \\
SG-CABOBA & $\logo(\varepsilon^{\textcolor{red}{-6}})$ &  This work, Theorem~\ref{thm:CABOBA_SG}  & $\wsone$ & Bregman, Gauge \\
SG-CABOBA & $\logo(\varepsilon^{\textcolor{red}{-(9p+6)/(p+2)}})$ &  This work, Theorem~\ref{thm:CABOBA_SG}  & $\wstwo$ & Bregman, Gauge \\
SG-CKLMC & $\logo(\varepsilon^{-18p/(p+2)})$ & This work, Theorem~\ref{thm:CKLMC_SG}  & $\wstwo$ & Bregman, Gauge \\
\hline
\end{tabular}
\end{center}
\caption{Iteration complexities (number of gradient evaluations for the stochastic gradient versions, counting a minibatch of size $b$ as $b$ evaluations) required to achieve an $\varepsilon$-accurate approximation in the Wasserstein metric, for fixed dimension $p>2$. All bounds for the splitting schemes are obtained under strong convexity and smoothness only (Assumptions~\ref{asm:radius}--\ref{asm:smooth}, and Assumption~\ref{asm:stoch_grad} for the stochastic gradient versions), no assumption on third derivatives is made. As $p\to\infty$, the $\wstwo$ exponents $(7p+2)/(p+2)$, $(9p+6)/(p+2)$ and $18p/(p+2)$ approach $7$, $9$ and $18$. The stochastic gradient splitting schemes use batch sizes $b=\Theta(\varepsilon^{-2})$ (Corollaries~\ref{cor:sg-cbaoab} and \ref{cor:sg-caboba}), compared with $\varepsilon^{-8p/(p+2)}$ and $\varepsilon^{-4}$ for SG-CKLMC (Corollary~\ref{cor:CKLMC_SG}). The improvements over SG-CKLMC are highlighted in red.}
\label{table:summary}
\end{table}
%{ the CKLMC and CLMC rows are quoted from \cite{YuYu2,Mert24}. Please double-check that they are computed with the same accounting as ours (same penalty-error rate from Proposition~\ref{prop:w2_tight}, same metric, and the step-size restriction $h\lesssim m/(\gamma M^\lambda)$ of Euler-type schemes), since under the accounting used here the $\wstwo$ exponent of CKLMC would depend on $p$.}
 
\subsection{Outline}
The outline of this paper is as follows. In section~\ref{sec:back} we provide a summary of our assumptions, and setup, while  introducing constrained splitting schemes. Section~\ref{sec:main} extends this by presenting our main theoretical findings, which include contraction, bias and complexity bounds. All of the proofs of our results are deferred to Appendix~\ref{sec:app_theory}. Numerical simulations are given in Section~\ref{sec:num}, which highlight the benefit of the constrained algorithms. Finally we conclude in Section~\ref{sec:conc}. The algorithms are stated in Appendix~\ref{sec:app_alg}.
 
\subsection{Notation}
\label{sec:not}
Denote the $p$-dimensional Euclidean space by $\R^p$.
We use $\btheta$ for deterministic vectors and
$\bvartheta$ for random vectors.
For symmetric matrices $\bfA$
and $\bfB$, we write $\bfA
\preccurlyeq\bfB$ (or $\bfB\succcurlyeq \bfA$) if $\bfB - \bfA$ is
positive semi-definite.
For a measurable function $f:\R^p\to\R$ and a set $\K\subset\R^p$, define the oscillation
$
\operatorname{osc}_{\K}(f):=\sup_{x\in\K} f(x)-\inf_{x\in\K} f(x).
$
For a twice differentiable function $f$, we denote its gradient and Hessian by $\nabla f$ and $\nabla^2 f$, respectively.
$\delta_{x}$ denotes the Dirac measure concentrated at the point $x$. For a random vector $X$ we write $\|X\|_{\Ltwo}:=(\E\|X\|^2)^{1/2}$.
For non-negative functions $f$ and $g$, we write $f(\varepsilon) = \mathcal{O}(g(\varepsilon))$ if there exists a constant $C > 0$ such that $f(\varepsilon) \leqslant C g(\varepsilon)$ for all sufficiently small $\varepsilon > 0$. Similarly, we write $f(\varepsilon) = \Omega(g(\varepsilon))$ if there exists a constant $c > 0$ such that $f(\varepsilon) \geqslant c g(\varepsilon)$ for all sufficiently small $\varepsilon > 0$. In the corollaries, statements such as
$\lambda=\Theta(\varepsilon^{2p/(p+2)})$ prescribe a choice of the tuning parameter
accuracy $\varepsilon$. The corresponding complexity bounds are then stated using
$\widetilde{\mathcal{O}}$ and $\widetilde{\Omega}$, which suppress log factors. The convention $a \wedge b := \min\{a,b\}$ is used throughout.
 
\section{Background Material \& Algorithms}
\label{sec:back}
In this section we provide the necessary background material before discussing our main results. This includes the assumptions on the potential function and on the compact set $\mathcal{K}$, and the introduction of the constrained algorithms CBAOAB and CABOBA and their stochastic gradient versions. All algorithms are expressed in algorithmic form in the appendices. We begin this section by assuming the convex and compact set $\K$ satisfies the following  assumption.
 
\begin{assumption}\label{asm:radius}
Given constants $0<r<R<\infty$, it holds that $B_2(r)\subset \K\subset B_2(R)$, where
$B_2(r)$ denotes the Euclidean ball of radius $r$ centered at the origin.
\end{assumption}
Moreover, we impose the following assumptions on the functions $f$ and $d_{\K}$.
\begin{assumption}\label{asm:f}
The function $f: \mathbb R^p\to\mathbb R$ attains its global minimum at the origin.
\end{assumption}
 
\begin{assumption}
\label{asm:dK}
There exists some constants $0<c_1\leqslant c_2$ such that $$c_1\|\btheta-\proj_{\K}(\btheta)\|_2^{2}\leqslant d_{\K}(\btheta)\leqslant c_2\|\btheta-\proj_{\K}(\btheta)\|_2^{2},$$ where $\proj_{\K}:\R^p\to\R^p$ is a projection operator onto set $\K.$
Moreover, $d_{\K}(\btheta)\geqslant 0$ for all $\btheta\in\mathbb{R}^p$, and $d_{\K}(\btheta)=0$ whenever $\btheta\in\K.$
\end{assumption}
 
Assumption \ref{asm:radius} has been made in various works in the area of constrained sampling \cite{pereyra,Mert24}. Assumption \ref{asm:f} is a technical assumption that the global minimum is at the origin. Finally,
Assumption \ref{asm:dK} is an assumption on the distance function, which ensures the potential $U^{\lambda}$ inherits smoothness and strong convexity from $f$.
 
\begin{remark}
Assumption \ref{asm:f} is without loss of generality, up to a translation of $\K$. Since $f$ is $m$-strongly convex, it has a unique minimizer $\btheta^*$. The change of variables $\widetilde{\btheta} = \btheta - \btheta^*$ shifts the minimizer to the origin while preserving $m$, $L$, convexity of $\mathcal{K}$, and Wasserstein distances.  Assumption~\ref{asm:radius} is then required for the translated set.
\end{remark}
 
Related to Assumption \ref{asm:dK}, we now introduce two choices of $\ell_\K$, based on commonly used projection classes, which are the Bregman and the Gauge projection. These projections have been widely used in various machine learning applications, which include clustering and detection of outliers \cite{class,detect}.
 
\begin{example}[Bregman projection]
\label{ex:breg}
Consider the Bregman projection $\proj^B_{\mathcal{K}}: \mathbb{R}^p \rightarrow \mathcal{K}$, associated with a fixed positive semidefinite matrix $Q \in \mathbb{R}^{p\times p}$, defined as
$$
\proj^B_{\mathcal{K}}(\btheta) : = \argmin_{\btheta' \in \mathcal{K}}(\btheta-\btheta')^{\top}Q(\btheta-\btheta').
$$
\end{example}
 
\begin{example}[Scaling-based projection]
Consider the Gauge projection $\proj^G_{\mathcal{K}}: \mathbb{R}^p \rightarrow \mathcal{K}$, defined as
$
\proj^G_{\mathcal{K}}(\btheta) : = {\btheta}/{g_{\mathcal{K}}(\btheta)},
$
where $g_{\mathcal{K}}$ is a variant of the Gauge function associated with the set $\mathcal{K}$, given as
$
g_{\mathcal{K}}(\btheta) := \inf \{t \geq 1: \btheta \in t\mathcal{K}\}.
$
\end{example}
 
We remark that the Bregman distance is equivalent to the squared Mahalanobis distance, defined as $F(\btheta) = \frac{1}{2}\btheta^{\top}Q\btheta$. Therefore, we can consider a distance function of the form
$
d_{\mathcal{K}}(\btheta) = \| \btheta - \proj^B_{\mathcal{K}}(\btheta) \|^2_2.
$
 
Let us now place common assumptions on $f$, which can be found below. They are generic assumptions used within the literature of sampling as an optimization procedure \cite{Chewi2025}.
 
\begin{assumption}
\label{asm:smooth}
The function $f:\mathbb{R}^p\to\mathbb{R}$ is  twice continuously differentiable, and $m$-strongly convex, that is,
$$
f(\btheta') \geqslant f(\btheta) + \langle \nabla f(\btheta),\btheta'-\btheta \rangle + \frac{m}{2}\|\btheta-\btheta'\|^2.
$$
Moreover, $f$ is $L$-smooth, which satisfies the following
\begin{align*}
\| \nabla f(\btheta) - \nabla f(\btheta') \| \leqslant L \| \btheta - \btheta' \|, \quad \forall \btheta, \btheta' \in \mathbb{R}^p,
\end{align*}
where $m>0$ and $L>0$ are the strong convexity and Lipschitz constants.
\end{assumption}
As shown in Lemmas 12 and 13 from~\cite{YuYu2}, the surrogate potential $U^\lambda$ inherits the smoothness and strong convexity of $f$, provided that the distance function $d_{\K}$ satisfies Assumption~\ref{asm:dK}.
More specifically, $U^\lambda$ is $m$-strongly convex and $M^\lambda$-smooth with
\begin{align*}
    M^\lambda = L + \frac{1}{\lambda^2} C_{\mathrm{proj}},
\end{align*}
where the constant $C_{\mathrm{proj}}>0$ depends on the specific choice of the projection operator.
In this work, we focus on the Bregman projection and the Gauge projection introduced in~\cite{YuYu2}, for which $C_{\mathrm{proj}}$ can be taken as a universal positive constant. In particular $M^\lambda=\bigo(1/\lambda^2)$ as $\lambda\to 0$.
 
\begin{remark}
\label{rem:regularity}
The surrogate potential $U^\lambda$ is continuously differentiable with Lipschitz gradient, but it is in general \emph{not} twice continuously differentiable. For instance, for the Euclidean projection ($Q=I$ in Example~\ref{ex:breg}) and a convex set $\K$ with smooth boundary, $\nabla d_\K(\btheta)=2(\btheta-\proj_\K(\btheta))$, so that $\nabla^2 d_\K=0$ in the interior of $\K$, while for $\btheta\notin\K$ close to a boundary point with outer unit normal $\mathbf n$ one has $\nabla^2 d_\K(\btheta)\to 2\,\mathbf n\mathbf n^\top$. Hence the Hessian of $U^\lambda$ jumps by a matrix of operator norm $1/\lambda^2$ across $\partial\K$, and $U^\lambda$ does not have a Lipschitz (or even continuous) Hessian. The same phenomenon occurs for the Gauge penalty. For this reason we do not impose any assumption on the third derivatives of $U^\lambda$. All the results below only use $m$-strong convexity and $M^\lambda$-smoothness. This excludes second-order bias estimates (which require a Lipschitz Hessian, see \cite[Theorem 8.5]{leimkuhler2023contractiona}) as well as the dimension-dependent analysis of UBU in \cite{sanz2021wasserstein,paulin2024correction}.
\end{remark}

\begin{remark}
\label{rem:satisfiable}
Assumptions~\ref{asm:radius}--\ref{asm:smooth} are simultaneously satisfiable
by the squared-distance penalty $d_{\kappa}(\theta)=\|\theta-\mathsf{P}_{\kappa}(\theta)\|_{2}^{2}$. With this choice, Assumption~\ref{asm:dK} holds with $c_{1}=c_{2}=1$, and, as recorded
above, $U^{\lambda}$ is $m$-strongly convex and $M^{\lambda}$-smooth with
$M^{\lambda}=L+\lambda^{-2}C_{\mathrm{proj}}$.
\end{remark}
 
\subsection{Constrained Splitting Schemes}
 
Throughout, $\mathcal{B}_t,\mathcal{A}_t,\mathcal{O}_t$ denote the maps \eqref{eq:BAO} with $U$ replaced by $U^\lambda$, that is, $\mathcal{B}_t(\bvartheta,\bv)=(\bvartheta,\bv-t\nabla U^\lambda(\bvartheta))$. We write a composition of maps by listing them in the order in which they are applied, and use the friction factor $\etah:=e^{-\gamma h}$.
 
\paragraph{CBAOAB.} One step of \textit{constrained-BAOAB} with step size $h>0$ is
\begin{equation}
\label{eq:CBAOAB}
\mathcal{B}_{h/2}\circ\mathcal{A}_{h/2}\circ\mathcal{O}_{h}\circ\mathcal{A}_{h/2}\circ\mathcal{B}_{h/2}.
\end{equation}
Since the gradient computed in the final $\mathcal B_{h/2}$ of a step can be reused in the first $\mathcal B_{h/2}$ of the next step, CBAOAB requires one gradient evaluation per step. Algorithm~\ref{alg:SG-CBAOAB} in the appendix gives the algorithmic form.
 
\paragraph{CABOBA.} One step of \textit{constrained-ABOBA} with step size $h>0$ is
\begin{equation}
\label{eq:CABOBA}
\mathcal{A}_{h/2}\circ\mathcal{B}_{h/2}\circ\mathcal{O}_{h}\circ\mathcal{B}_{h/2}\circ\mathcal{A}_{h/2}.
\end{equation}
Since the position is not updated between the two $\mathcal B_{h/2}$ maps, both use the gradient at the same point $\bvartheta+\frac h2\bv$, and CABOBA requires exactly one gradient evaluation per step. Algorithm~\ref{alg:SG-CABOBA} in the appendix gives the algorithmic form.
 
\begin{remark}
Other symmetric BAO splittings, such as OBABO, can be treated with the same techniques. Our reason for considering BAOAB and ABOBA is that they are the two splittings in which the stochastic $\mathcal O$ step sits between the deterministic parts. In particular BAOAB is exact for Gaussian targets \cite{leimkuhler2013jcp} and has the $\gamma$-limit convergent property of \cite{leimkuhler2023contractiona}. ABOBA additionally evaluates both half kicks at the same point, which yields a particularly simple structure of the stochastic-gradient error (see Section~\ref{sec:sg-caboba}).
\end{remark}

\begin{remark}
\label{rem:no-ubu}
In this work we do not consider the UBU or SG-UBU schemes \cite{sanz2021wasserstein,supp:UBUBU}, despite its strong
order two, since its Wasserstein analysis \cite{sanz2021wasserstein,paulin2024correction}
requires a bounded third derivative of the potential. As explained in
Remark~\ref{rem:regularity}, the Hessian of $U^\lambda$ jumps across $\partial\K$,
so this assumption cannot hold in the penalized setting.
\end{remark}
 
\subsection{Extension to Stochastic Gradients}
For many practical scenarios, the exact computation of the gradient can be time-consuming, especially for large datasets. As a result, we consider the use of stochastic gradients (SG) for the potential $f$, based on mini-batching, while the gradient of the penalty $\ell_\K^\lambda$, which does not depend on the data, is computed exactly.
\begin{definition} \label{def:stochastic_gradient}
A \textit{stochastic gradient approximation} of a potential $\uu$ is defined by a function $\mathcal{G}:\R^p \times \Omega \to \R^p$ and a probability distribution $\rho$ on a Polish space $\Omega$, such that for every $\btheta\in \R^p$, $\mathcal{G}(\btheta, \cdot)$ is measurable on $(\Omega,\mathcal{F})$, and for $\omega\sim \rho$,
$\E(\mathcal{G}(\btheta,\omega)) = \nabla \uu(\btheta).$
The function $\mathcal{G}$ and the distribution $\rho$  together define the stochastic gradient, which we denote as  $(\mathcal{G}, \rho)$.
\end{definition}
The stochastic gradient of the surrogate potential is then $\widetilde\nabla U^\lambda(\btheta,\omega):=\mathcal G(\btheta,\omega)+\nabla\ell^\lambda_\K(\btheta)$.
 
\begin{assumption}
\label{asm:stoch_grad}
The stochastic gradient $(\mathcal{G}, \rho)$ of Definition \ref{def:stochastic_gradient} satisfies:
\begin{itemize}
\item[(ii)] $\E_{\omega\sim\rho}\big[ \|\mathcal{G}(\btheta,\omega) - \nabla f(\btheta)\|^2 \big]\leqslant \sigma^2_1L^2\|\btheta\|^2$ for all $\btheta\in\R^p$ (bound on the variance).
\end{itemize}
\end{assumption}
In practice, the potential function $f$ often takes the form
$
f(\btheta)=\sum_{i=1}^{N_D} f_i(\btheta)\,.
$
Accordingly, the stochastic gradient of $f$ is typically computed as
$
\mathcal G(\btheta,\omega)=\frac{N_D}{b}\sum_{j\in \omega}\nabla f_j(\btheta),
$
where $\omega\subset\{1,\dots,N_D\}$ is a mini-batch of size $b$.
In this setting $\sigma_1^2=\bigo(1/b)$ under standard conditions on the individual $f_j$.
In SG-CBAOAB (Algorithm \ref{alg:SG-CBAOAB}) a fresh stochastic gradient is drawn at each new position and reused for the two half kicks surrounding it, and in SG-CABOBA (Algorithm \ref{alg:SG-CABOBA}) one stochastic gradient is drawn per step and used for both half kicks. Both algorithms thus use one stochastic gradient per step.
 
\subsection{Convergence in Wasserstein Distance}
\label{subsec:wass}
 
We will consider the Wasserstein distance, which is a popular metric to show convergence of sampling algorithms \cite{Eberle2019, dalalyan2017theoretical}. The key underlying idea is that if one shows contraction between two measures, then this implies a unique invariant measure and convergence towards it. Our error bounds are based on the decomposition
$$
\wsq(\nu_n, \nu) \leq \underbrace{\wsq(\nu_n, \nu^\lambda)}_{\mathrm{sampling \ error}}+\underbrace{\wsq(\nu^\lambda, \nu)}_{\mathrm{approximation \ error}},
$$
where $\nu_n$ denotes the law of the position of the $n$-th iterate of an algorithm.
 
\begin{definition}[$q$-Wasserstein distance]
\label{def:wass}
For a norm $\|\cdot\|_\star$ on $\R^d$ and $q\in[1,\infty)$, let $\mathcal{P}_q(\R^{d})$ denote the probability measures with finite $q$-th moment. The $q$-Wasserstein distance between $\mu,\mu'\in \mathcal{P}_q(\R^{d})$ with respect to $\|\cdot\|_\star$ is
\begin{equation}
\label{eq:wass_dist}
{\sf W}_{q,\star}(\mu,\mu') =\Big(\inf_{\xi \in \Gamma(\mu,\mu')} \int_{\R^{d} \times \R^{d}} \|z_1 - z_2\|^q_{\star}\,d\xi(z_1,z_2)\Big)^{1/q},
\end{equation}
where $\Gamma(\mu,\mu')$ is the set of couplings of $\mu$ and $\mu'$. When $\|\cdot\|_\star$ is the Euclidean norm we write $\wsq$.
\end{definition}
 
We now state an upper bound for the approximation error, taken from \cite[Proposition 1]{YuYu2}, which shows that $\wsq(\nu,\nu^\lambda) \rightarrow 0$ as $\lambda \rightarrow 0$.
 
\begin{proposition}
\label{prop:w2_tight}
Under Assumptions~\ref{asm:radius}-\ref{asm:dK},
for any $q\geqslant 1$ and any $\lambda$ satisfying
\begin{align*}
0< \lambda <\Bigl( \frac{\sqrt{c_1}\,r}{p+q} \bigwedge \frac{\sqrt{c_1}\,r\,e^{\operatorname{osc}_{\K}(f)}}{3\sqrt{\pi}\,\operatorname{Vol}(\K)\,p}\Bigr),
\end{align*}
it holds that
\begin{align*}
    \wsq(\nu,\nu^\lambda)\leqslant C(p,q)\cdot\begin{cases}
        \lambda, &1/p+1/q>1,\\
        \lambda\log^{1/q}(\frac{1}{\lambda}), &1/p+1/q=1,\\
        \lambda^{1/p+1/q}, &1/p+1/q<1
    \end{cases}
\end{align*}
where $$C(p,q)=C_0\left[p \cdot \frac{e^{4\operatorname{osc}_{\K}(f)}}{c_1^{1/2}}\cdot\left(\frac{\max(R,1)}{\min(r,1)}\right)^{(2p+1)}\right]^{1/p+1/q},$$ and $C_0>0$ is a universal constant.
\end{proposition}
 
\section{Main Results}
\label{sec:main}
In this section we present our main results. Section~\ref{sec:prelim} contains the common ingredients, which are a shifted weighted norm, one-step contraction of both schemes, and the notation for the bias. Sections~\ref{sec:cbaoab}--\ref{sec:sg-caboba} contain the non-asymptotic bounds and complexity estimates for CBAOAB, SG-CBAOAB, CABOBA and SG-CABOBA, and Section~\ref{sec:sgcklmc} the results for SG-CKLMC. All proofs are deferred to Appendix~\ref{sec:app_theory}.
 
Throughout this section we write $M:=M^\lambda$, we let
$$
\pi^\lambda:=\nu^\lambda\otimes\mathcal N(0,I_p),
$$
denote the invariant measure of the KLD \eqref{eq:kinetic_langevin} with potential $U^\lambda$, and for a measure $\mu$ on $\R^{2p}$ we denote by $\mu^{\mathsf x}$ its position marginal. All Wasserstein distances between a law $\nu_n$ of iterates and $\nu$ or $\nu^\lambda$ refer to the position marginal.
 
\subsection{Preliminaries: a shifted norm and one-step contraction}
\label{sec:prelim}
 
Following \cite{leimkuhler2023contractiona}, for $a,b>0$ define the weighted Euclidean norm on $\R^{2p}$
\[
\|(x,v)\|_{a,b}^2:=\|x\|^2+2b\langle x,v\rangle+a\|v\|^2 ,
\]
and fix the parameters
\begin{equation}
\label{eq:abc}
a=\frac1{M},\qquad b=\frac{h}{1-\etah},\qquad c(h)=\frac{h^2m}{4(1-\etah)} .
\end{equation}
We also define the linear shift $T(x,v):=(x+\tfrac h2 v,v)$ and the \emph{shifted norm}
\begin{equation}
\label{eq:sharp}
\nsharp{z}:=\|T^{-1}z\|_{a,b},\qquad\text{i.e.}\qquad \nsharp{(x,v)}^2=\Big\|x-\tfrac h2 v\Big\|^2+2b\Big\langle x-\tfrac h2 v,v\Big\rangle+a\|v\|^2 .
\end{equation}
We denote by $\wsharp$ the $2$-Wasserstein distance on $\R^{2p}$ with respect to $\nsharp{\cdot}$. Our standing step-size condition is
\begin{equation}
\label{eq:stepsize}
h<\frac{1-e^{-\gamma h}}{2\sqrt{M}} ,
\end{equation}
which is the condition of \cite[Theorems 5.1 and 5.2]{leimkuhler2023contractiona}. It implies $h\sqrt M<1/2$, and it holds for instance when $\gamma=4\sqrt M$ and $h\leqslant 1/(4\sqrt M)$. See the proof of Corollary~\ref{cor:cbaoab}.
 
To state contraction for CBAOAB we use the interior chain, if $(\bvartheta_k,\bv_k)$ are the CBAOAB iterates, let
\begin{equation}
\label{eq:interior}
\bz_k:=\mathcal B_{-h/2}(\bvartheta_k,\bv_k)=\big(\bvartheta_k,\bv_k+\tfrac h2\nabla U^\lambda(\bvartheta_k)\big).
\end{equation}
Since consecutive half kicks merge, $\bz_{k+1}$ is obtained from $\bz_k$ by applying $\mathcal B_h\mathcal A_{h/2}\mathcal O_h\mathcal A_{h/2}$, and the position of $\bz_k$ coincides with $\bvartheta_k$. We now provide a proposition demonstrating the one-step contraction.
 
\begin{proposition}
\label{prop:contraction}
Suppose $U^\lambda$ is $m$-strongly convex and $M$-smooth and that \eqref{eq:stepsize} holds. Let $\Phi(z,\bxi)$ denote either one CABOBA step \eqref{eq:CABOBA} or one step $\mathcal B_h\mathcal A_{h/2}\mathcal O_h\mathcal A_{h/2}$ of the CBAOAB interior chain, from $z\in\R^{2p}$ with Gaussian increment $\bxi$. Then for all $z,z'\in\R^{2p}$ and all $\bxi$,
\[
\nsharp{\Phi(z,\bxi)-\Phi(z',\bxi)}\leqslant (1-c(h))^{1/2}\,\nsharp{z-z'} .
\]
Consequently, for all $\mu,\mu'\in\mathcal P_2(\R^{2p})$ and $n\in\mathbb N$, $\wsharp(\mu P_h^n,\mu' P_h^n)\leqslant (1-c(h))^{n/2}\wsharp(\mu,\mu')$, where $P_h$ is the corresponding Markov kernel, and $P_h$ has a unique invariant measure $\pi_h\in\mathcal P_2(\R^{2p})$.
\end{proposition}
 
The proof shows that both chains are conjugate, through the linear map $T$, to chains for which one-step contraction in $\|\cdot\|_{a,b}$ is established in the proofs of \cite[Theorems 5.1 and 5.2]{leimkuhler2023contractiona}, and uses a mean value theorem for $C^{1,1}$ functions to remove the twice-differentiability required there. Compared with the BAOAB result of \cite[Theorem 5.1]{leimkuhler2023contractiona}, the shifted norm removes the multiplicative prefactor, which is convenient for the stochastic gradient analysis.
 
The bias of both schemes will be expressed through
\begin{equation}
\label{eq:Bh}
\mathsf B_h:=\frac{(1-e^{-\gamma h})\sqrt{M p}}{m}\ \leqslant\ \frac{\gamma h\sqrt{Mp}}{m}.
\end{equation}
 
\subsection{Convergence Analysis of CBAOAB}
\label{sec:cbaoab}
 We now present out first result of our splitting constrained schemes, specifically convergence for CBAOAB.
\begin{theorem}
\label{thm:cbaoab}
Under Assumptions~\ref{asm:radius}--\ref{asm:smooth}, assume that $\lambda$ satisfies the condition of Proposition~\ref{prop:w2_tight} with $q=2$, that $p>2$, and that $\gamma,h>0$ satisfy \eqref{eq:stepsize}. Let $\nu_n$ denote the law of the $n$-th position iterate of CBAOAB initialized at $\bvartheta_0=\mathbf 0$, $\bv_0\sim\mathcal N(0,I_p)$. Then, for all $n\in\mathbb N$,
\begin{align*}
\wstwo(\nu_n,\nu)
&\leqslant (1-c(h))^{n/2}\Big(\frac52\sqrt{\frac pm}+8\,\mathsf B_h\Big)+8\,\mathsf B_h+C(p,2)\lambda^{\frac12+\frac1p},\\
\wsone(\nu_n,\nu)
&\leqslant (1-c(h))^{n/2}\Big(\frac52\sqrt{\frac pm}+8\,\mathsf B_h\Big)+8\,\mathsf B_h+C(p,1)\lambda,
\end{align*}
where $c(h)$ is given in \eqref{eq:abc}, $\mathsf B_h$ in \eqref{eq:Bh}, and $C(p,q)$ in Proposition~\ref{prop:w2_tight}. Moreover $(1-c(h))^{n/2}\leqslant e^{-mhn/(8\gamma)}$.
\end{theorem}
 
The term $8\mathsf B_h$ bounds the asymptotic bias $\wstwo(\pi_h^{\mathsf x},\nu^\lambda)$ of the scheme (Proposition~\ref{prop:bias-cbaoab}). It is of first order in $h$, consistently with \cite[Theorem 8.5]{leimkuhler2023contractiona} under the same assumptions, but with a much smaller constant.
 
\begin{corollary}
\label{cor:cbaoab}
Let the error level $\varepsilon\in(0,1)$ be small, and set $\gamma=4\sqrt{M^\lambda}$.
\begin{itemize}
\item[(a)] Set $\lambda=\Theta(\varepsilon^{\frac{2p}{p+2}})$ and choose $h>0$ and $n\in\mathbb N_+$ so that
\[
h=\Theta\big(\varepsilon^{\frac{5p+2}{p+2}}\big),\qquad n=\widetilde\Omega\big(\varepsilon^{-\frac{7p+2}{p+2}}\big),
\]
then $\wstwo(\nu_n,\nu)=\bigo(\varepsilon)$.
\item[(b)] Set $\lambda=\Theta(\varepsilon)$ and choose $h>0$ and $n\in\mathbb N_+$ so that
\[
h=\Theta(\varepsilon^{3}),\qquad n=\widetilde\Omega(\varepsilon^{-4}),
\]
then $\wsone(\nu_n,\nu)=\bigo(\varepsilon)$.
\end{itemize}
\end{corollary}
 
\subsection{Extension to SG-CBAOAB}
\label{sec:sg-cbaoab}
 
For the stochastic gradient counterpart, we adopt a strategy of comparing the stochastic gradient chain with the full gradient chain started from the same initial condition. We then bound the one-step discrepancy between the two kernels, and summing these discrepancies using the contraction of the full gradient kernel. For CBAOAB we work with the interior chain \eqref{eq:interior}, whose position coincides with the CBAOAB iterate.
 
\begin{theorem}
\label{thm:CBAOAB_SG}
Under the assumptions of Theorem~\ref{thm:cbaoab} and Assumption~\ref{asm:stoch_grad}, suppose in addition that
\begin{equation}
\label{eq:small-noise-baoab}
\sigma_1L\leqslant \frac{h\,m\sqrt{M}}{38\,(1-e^{-\gamma h})}.
\end{equation}
Let $\nu_n$ denote the law of the $n$-th position iterate of SG-CBAOAB (Algorithm~\ref{alg:SG-CBAOAB}) initialized at $\bvartheta_0=\mathbf 0$, $\bv_0\sim\mathcal N(0,I_p)$. Then, for all $n\in\mathbb N$,
\begin{align*}
\wstwo(\nu_n,\nu)
&\leqslant (1-c(h))^{n/2}\Big(\frac52\sqrt{\frac pm}+8\,\mathsf B_h\Big)
+\frac{16\,(1-e^{-\gamma h})}{h\,m\sqrt{M}}\,\sigma_1 L\,\mathsf S_{\flat}
+8\,\mathsf B_h+C(p,2)\lambda^{\frac12+\frac1p},
\end{align*}
where $\mathsf S_\flat:=\frac{11}2\sqrt{p/m}+28\,\mathsf B_h$, and the same bound holds for $\wsone(\nu_n,\nu)$ with $C(p,2)\lambda^{\frac12+\frac1p}$ replaced by $C(p,1)\lambda$. Moreover $(1-c(h))^{n/2}\leqslant e^{-mhn/(8\gamma)}$.
\end{theorem}
 
The stochastic gradient contribution is of order $(1-e^{-\gamma h})\sigma_1/(hm\sqrt M)\leqslant\gamma\sigma_1/(m\sqrt M)$. The gradient noise perturbs each step by an amount of order $h\sigma_1/\sqrt M$ in the norm \eqref{eq:sharp}, and these perturbations are summed over the $\mathcal O(1/c(h))$ steps of a mixing time. This contribution does not vanish as $h\to0$, so the variance $\sigma_1^2$ has to be of order $\varepsilon^2$, which with $\sigma_1^2=\bigo(1/b)$ requires batch sizes of order $\varepsilon^{-2}$.
 
\begin{corollary}
\label{cor:sg-cbaoab}
Let the error level $\varepsilon\in(0,1)$ be small, set $\gamma=4\sqrt{M^\lambda}$, and assume $\sigma_1^2=\bigo(1/b)$. Choose $\lambda$, $h$ and $n$ as in Corollary~\ref{cor:cbaoab}(a) (resp.\ (b)) and the batch size $b=\Theta(\varepsilon^{-2})$. Then $\wstwo(\nu_n,\nu)=\bigo(\varepsilon)$ (resp.\ $\wsone(\nu_n,\nu)=\bigo(\varepsilon)$), and the number of stochastic gradient evaluations $nb$ is $\widetilde\bigo(\varepsilon^{-\frac{9p+6}{p+2}})$ for $\wstwo$ and $\widetilde\bigo(\varepsilon^{-6})$ for $\wsone$.
\end{corollary}
 
\subsection{Convergence Analysis of CABOBA}
\label{sec:caboba}
 
\begin{theorem}
\label{thm:caboba}
Under Assumptions~\ref{asm:radius}--\ref{asm:smooth}, assume that $\lambda$ satisfies the condition of Proposition~\ref{prop:w2_tight} with $q=2$, that $p>2$, and that $\gamma,h>0$ satisfy \eqref{eq:stepsize}. Let $\nu_n$ denote the law of the $n$-th position iterate of CABOBA initialized at $\bvartheta_0=\mathbf 0$, $\bv_0\sim\mathcal N(0,I_p)$. Then, for all $n\in\mathbb N$,
\begin{align*}
\wstwo(\nu_n,\nu)
&\leqslant (1-c(h))^{n/2}\Big(\frac32\sqrt{\frac pm}+\frac{36}{5}\,\mathsf B_h\Big)+\frac{36}{5}\,\mathsf B_h+C(p,2)\lambda^{\frac12+\frac1p},\\
\wsone(\nu_n,\nu)
&\leqslant (1-c(h))^{n/2}\Big(\frac32\sqrt{\frac pm}+\frac{36}{5}\,\mathsf B_h\Big)+\frac{36}{5}\,\mathsf B_h+C(p,1)\lambda,
\end{align*}
and $(1-c(h))^{n/2}\leqslant e^{-mhn/(8\gamma)}$.
\end{theorem}
 
\begin{corollary}
\label{cor:caboba}
The conclusions of Corollary \ref{cor:cbaoab} hold for CABOBA, with the same
choices of $\lambda$, $h$ and $n$.
\end{corollary}
 
\subsection{Extension to SG-CABOBA}
\label{sec:sg-caboba}
 
For SG-CABOBA, the single stochastic gradient of a step is evaluated at $\widetilde\bvartheta_k:=\bvartheta_k+\frac h2\bv_k$ and used in both half kicks. A short computation (Lemma~\ref{lem:sg-caboba}) shows that the resulting perturbation of the step has the form $(\frac h2\Delta,\Delta)$ for a velocity perturbation $\Delta$, so that it only affects the velocity coordinate in the shifted norm \eqref{eq:sharp}. The proof of Theorem~\ref{thm:CBAOAB_SG} then applies with minor modifications.
 
\begin{theorem}
\label{thm:CABOBA_SG}
Under the assumptions of Theorem~\ref{thm:caboba} and Assumption~\ref{asm:stoch_grad}, suppose in addition that
\begin{equation}
\label{eq:small-noise-aboba}
\sigma_1L\leqslant \frac{h\,m\sqrt{M}}{31\,(1-e^{-\gamma h})}.
\end{equation}
Let $\nu_n$ denote the law of the $n$-th position iterate of SG-CABOBA (Algorithm~\ref{alg:SG-CABOBA}) initialized at $\bvartheta_0=\mathbf 0$, $\bv_0\sim\mathcal N(0,I_p)$. Then, for all $n\in\mathbb N$,
\begin{align*}
\wstwo(\nu_n,\nu)
&\leqslant (1-c(h))^{n/2}\Big(\frac32\sqrt{\frac pm}+\frac{36}{5}\,\mathsf B_h\Big)
+\frac{12\,(1-e^{-\gamma h})}{h\,m\sqrt{M}}\,\sigma_1 L\,\mathsf S_{\sharp}
+\frac{36}{5}\,\mathsf B_h+C(p,2)\lambda^{\frac12+\frac1p},
\end{align*}
where $\mathsf S_\sharp:=\frac92\sqrt{p/m}+27\,\mathsf B_h$, and the same bound holds for $\wsone(\nu_n,\nu)$ with $C(p,2)\lambda^{\frac12+\frac1p}$ replaced by $C(p,1)\lambda$. Moreover $(1-c(h))^{n/2}\leqslant e^{-mhn/(8\gamma)}$.
\end{theorem}
 
\begin{corollary}
\label{cor:sg-caboba}
Under the assumptions of Theorem \ref{thm:CABOBA_SG}, the conclusions of
Corollary \ref{cor:sg-cbaoab} hold for SG-CABOBA, with the same choices of
$\lambda$, $h$, $n$ and $b$.
\end{corollary}
 
\begin{remark}
The full-gradient and stochastic-gradient bounds have the same structure for both schemes, and only differ in the constants. CABOBA has slightly smaller constants and a slightly weaker small-noise requirement, which is due to the fact that both of its half kicks are evaluated at the same point. With $\gamma=4\sqrt{M}$ and $h\leqslant1/(4\sqrt M)$, the right-hand sides of \eqref{eq:small-noise-baoab}--\eqref{eq:small-noise-aboba} are at least $m/152$, so these conditions only require the batch size to exceed a constant. In our numerical experiments (Section~\ref{sec:num-box}) much smaller batches than those of Corollaries~\ref{cor:sg-cbaoab}--\ref{cor:sg-caboba} already suffice, indicating that the bounds are conservative in their dependence on $\sigma_1$.
\end{remark}
 
\begin{remark}
If the penalty is replaced by a smoother one so that $U^\lambda$ has a Lipschitz Hessian, then second-order bias estimates for BAOAB \cite[Theorem 8.5]{leimkuhler2023contractiona} become available and the complexity improves accordingly. Since this would require a new approximation result in place of Proposition~\ref{prop:w2_tight}, we leave it for future work.
\end{remark}
 
\subsection{Results for SG-CKLMC}
\label{sec:sgcklmc}
 
Thanks to our refined analysis of the distance between the smooth surrogate and the target distribution, we are able to derive an upper bound on the convergence rate of CKLMC with stochastic gradients (referred to as SG-CKLMC), as stated in the following theorem.
\begin{theorem}\label{thm:CKLMC_SG}
Under Assumptions~\ref{asm:radius}–\ref{asm:smooth} and Assumption~\ref{asm:stoch_grad}, assume that $\lambda$ satisfies
\begin{align*}
0< \lambda < \Bigl( \frac{\sqrt{c_1}\,r}{p+2} \;\bigwedge\; \frac{\sqrt{c_1}\,r\,e^{\operatorname{osc}_{\K}(f)}}{3\sqrt{\pi}\,\operatorname{Vol}(\K)\,p}\Bigr).
\end{align*}
Then, for any $\gamma\geqslant \sqrt{m+M^\lambda},$ and any step size $h$ such that
\begin{align*}
0< h < \min\left\{\frac{\gamma
\tau }{2K_1},\frac{2}{\gamma \tau},\frac{1}{10\gamma},\frac{m}{4\gamma M^\lambda}\right\}\,,
\end{align*}
it holds that
\begin{align*}
\wstwo(\nu_n^{\sf SG-CKLMC},\nu)
& \leqslant \bigg(1-\frac{0.75mh}{\gamma}\bigg)^n  \wstwo(\nu_0^{\sf SG-CKLMC},\nu)
+ \frac{ \sqrt{2}M^\lambda h\sqrt{p}}{m}
 + 2C(p,2)\lambda^{1/2+1/p}
\\& + \frac{8\sqrt{2}\sigma_1 L}{m \gamma} \sqrt{\frac{C_{\mathcal{V}}}{1-2\tau}}\,.
\end{align*}
Here, the constants $\tau$ and $K_1$ are defined by
\begin{align*}
\tau
= \frac{1}{2}\min\left\{ \frac{1}{4}, \frac{m}{M^\lambda + \gamma^2/2} \right\}, \quad
K_1
= \max\left\{
\frac{16\big((M^\lambda)^2 + 2\gamma (M^\lambda)^2 + \sigma_1^2 L^2\big)}{(1-2\tau)\gamma^2},
\frac{4M^\lambda + 2\gamma^2(1-\tau) + 8\gamma}{1-2\tau}
\right\}.
\end{align*}
The constant $C_{\mathcal{V}}$ is given by
\[
C_{\mathcal{V}}=\int_{\mathbb{R}^{2p}}\mathcal{V}(\bvartheta,\bv)\, \mu_0(d\bvartheta,d\bv)+ \frac{4}{\tau}\left(p+\frac{mf(0)}{2M^\lambda +\gamma^2}\right)\,,
\]
where the Lyapunov function $\mathcal{V}$ is defined as
\[
\mathcal{V}(\bvartheta,\bv)=U^\lambda(\bvartheta)+\frac{\gamma^2}{4}(\|\bvartheta+\gamma^{-1}\bv\|^2+\|\gamma^{-1}\bv\|^2-\tau\|\bvartheta\|^2).
\]
 Moreover, when the initial point $\btheta_0^{\sf SG-CKLMC}=\mathbf{0}$ and $\bv_0^{\sf SG-CKLMC} \sim \mathcal{N}_p(0, I_p)$, it holds that
\begin{align*}
\wsone(\nu_n^{\sf SG-CKLMC},\nu)
\leqslant \bigg(1-\frac{0.75mh}{\gamma}\bigg)^n \sqrt{\frac{p}{m}}+\frac{ \sqrt{2}M^\lambda h\sqrt{p}}{m}
+ \frac{8\sqrt{2}\sigma_1 L}{m \gamma} \sqrt{\frac{C_{\mathcal{V}}}{1-2\tau}}
+ C(p,1)\lambda,
\end{align*}
where $\nu_n^{\mathsf{SG-CKLMC}}$ denotes the distribution of the $n$-th iterate of the SG-CKLMC scheme.
\end{theorem}
 
Adopting $M^\lambda=\bigo(1/\lambda^2)$ in the previously stated theorem, we obtain the following corollary.
\begin{corollary}
\label{cor:CKLMC_SG}
Let the error level $\varepsilon\in(0,1)$ be small.
\begin{itemize}
    \item[(a)] Set $\lambda=\Theta(h^{1/4})$, and choose $h>0$, the batch size and the number of iterations $b, n\in\mathbb N_+$ so that
\begin{align*}
    h=\mathcal O\left(\varepsilon^{\frac{8p}{p+2}}\right), \quad b=\Omega(\varepsilon^{-\frac{8p}{p+2}}), \quad n=\widetilde{\Omega}\left(\varepsilon^{\frac{-10p}{p+2}}\right),
\end{align*}
then we have $\wstwo(\nu_n^{\sf SG-CKLMC},\nu)=\logo(\varepsilon)$ after $nb$  stochastic gradient evaluations.
    \item[(b)] Set $\lambda=\Theta(h^{1/4})$, and choose $h>0$, the batch size and the number of iterations $b, n\in\mathbb N_+$ so that
\begin{align*}
    h=\mathcal O\left(\varepsilon^4\right),\quad b=\Omega(\varepsilon^{-4}), \quad n=\widetilde{\Omega}\left(\varepsilon^{-5}\right),
\end{align*}
then we have $\wsone(\nu_n^{\sf SG-CKLMC},\nu)=\logo(\varepsilon)$ after $nb$  stochastic gradient evaluations.
\end{itemize}
\end{corollary}
In comparison with Corollaries~\ref{cor:sg-cbaoab} and~\ref{cor:sg-caboba}, the bound for SG-CKLMC requires larger batch sizes. Specifically $\varepsilon^{-8p/(p+2)}$ and $\varepsilon^{-4}$ instead of $\varepsilon^{-2}$, and more iterations. This comes from the much smaller step sizes $h\lesssim m/(\gamma M^\lambda)$ allowed for the Euler-type discretization, compared with $h\lesssim1/\sqrt{M^\lambda}$ for the splitting schemes.
 
\begin{remark}
\label{rem:im}
The constrained randomized midpoint method (RMM) \cite{YuYu2} has better complexity guarantees than the splitting schemes considered here, the RMM being essentially optimal in the kinetic Langevin regime \cite{CLW20,lu2022explicit}. Our motivation for the splitting schemes is that (i) they require only one gradient evaluation per step, (ii) their contraction guarantees hold for step sizes of order $1/\sqrt{M^\lambda}$ (Proposition~\ref{prop:contraction}), rather than $m/(\gamma M^\lambda)$ for the Euler-type scheme in Theorem~\ref{thm:CKLMC_SG}, which is important for the stiff penalized potentials $U^\lambda$, and (iii) their stochastic gradient versions admit an analysis with substantially smaller batch sizes than SG-CKLMC.
\end{remark}
 
\section{Numerical Experiments}
\label{sec:num}
In this section we compare the full-gradient algorithms CKLMC, CBAOAB and CABOBA, and their stochastic gradient versions SG-CKLMC, SG-CBAOAB and SG-CABOBA. We consider two models; a toy circular constraint example, and a Bayesian constrained regression problem. In all experiments the penalty gradient is computed exactly and only $\nabla f$ is replaced by a minibatch gradient, all algorithms are run with the same $\lambda$, $\gamma$ and step size $h$, and we report the error of the law of the $n$-th iterate, estimated from a large number of independent chains, as in our theoretical results. Additional experiments for a triangular constraint, with modified step-sizes are provide in the Appendix.

% Step sizes are reported through the dimensionless quantity $h\sqrt{M^\lambda}$, since $M^\lambda$ governs the stiffness of the penalized problem. We use the friction $\gamma=2$ throughout: with $m=1$, larger friction slows down the mixing of the continuous dynamics, and although our theorems are stated under the step-size condition \eqref{eq:stepsize}, which requires $\gamma\gtrsim\sqrt{M^\lambda}$, the splitting schemes remain stable and accurate for $\gamma=2$ in all our experiments.
 
\subsection{Circular Constraint}
\label{sec:num-circle}
 
Let $f(\btheta)=\|\btheta\|^2/2$ on $\R^2$ and $\K=\mathcal B_2(0,R)$ with $R=0.5$, so that $\nu$ is a standard Gaussian restricted to the disc. For a centred ball the Gauge projection coincides with the Euclidean projection, and the penalty gradient is $\lambda^{-2}(1-R/\|\btheta\|)\btheta$ for $\|\btheta\|>R$. We take $\lambda=0.03$, so that $M^\lambda=1+\lambda^{-2}\approx1112$.
 
We write $f$ as a finite sum, $f(\btheta)=\frac1{N_D}\sum_{j=1}^{N_D}\frac12(a_j^\top\btheta)^2$ with $N_D=1000$ vectors $a_j\in\R^2$ normalized so that $\frac1{N_D}\sum_ja_ja_j^\top=I_2$, hence $f$ and $\nu$ are unchanged. The minibatch gradient with batch size $b=64$ is $\frac1b\sum_{j\in\omega}a_ja_j^\top\btheta$, whose variance is proportional to $\|\btheta\|^2$, as in Assumption~\ref{asm:stoch_grad}(ii).
 
All chains are started from $\bvartheta_0=\mathbf 0$, $\bv_0\sim\mathcal N(0,I_2)$ and run for a time horizon $T=8$, i.e.\ $n=\lceil T/h\rceil$ steps, using $10^4$ independent chains and 3 independent repetitions. Since $\nu$, $\nu^\lambda$, the initialization and the full-gradient dynamics are rotation invariant, the $\wstwo$ distance between the law of the iterates and $\nu^\lambda$ equals the one-dimensional $\wstwo$ distance between the laws of the radius $\|\bvartheta_n\|$, which we compute against the exact radial law of $\nu^\lambda$ obtained by quadrature. The penalty error is $\wstwo(\nu^\lambda,\nu)\approx0.027$, and the Monte Carlo resolution of our estimates is about $0.002$.
 
\begin{figure}[h!]
  \centering
  \includegraphics[width=\linewidth]{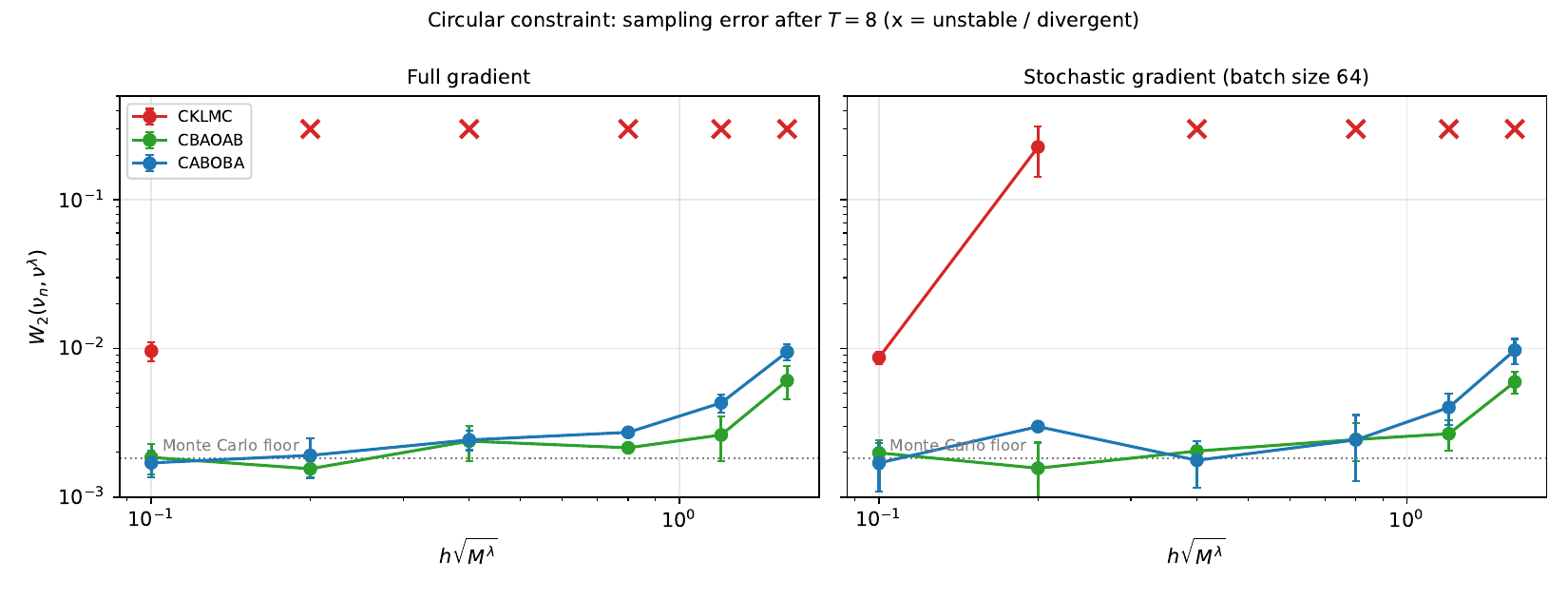}
  \caption{Circular constraint: sampling error $\wstwo(\nu_n,\nu^\lambda)$ after time $T=8$ as a function of the step size, for full gradients (left) and stochastic gradients with batch size $64$ (right). Error bars show the standard deviation over 3 repetitions. The crosses indicate step sizes at which the algorithm diverged in at least one repetition.}
\label{fig:circle-error}
\end{figure}
 
\begin{figure}[h!]
  \centering
  \includegraphics[width=\linewidth]{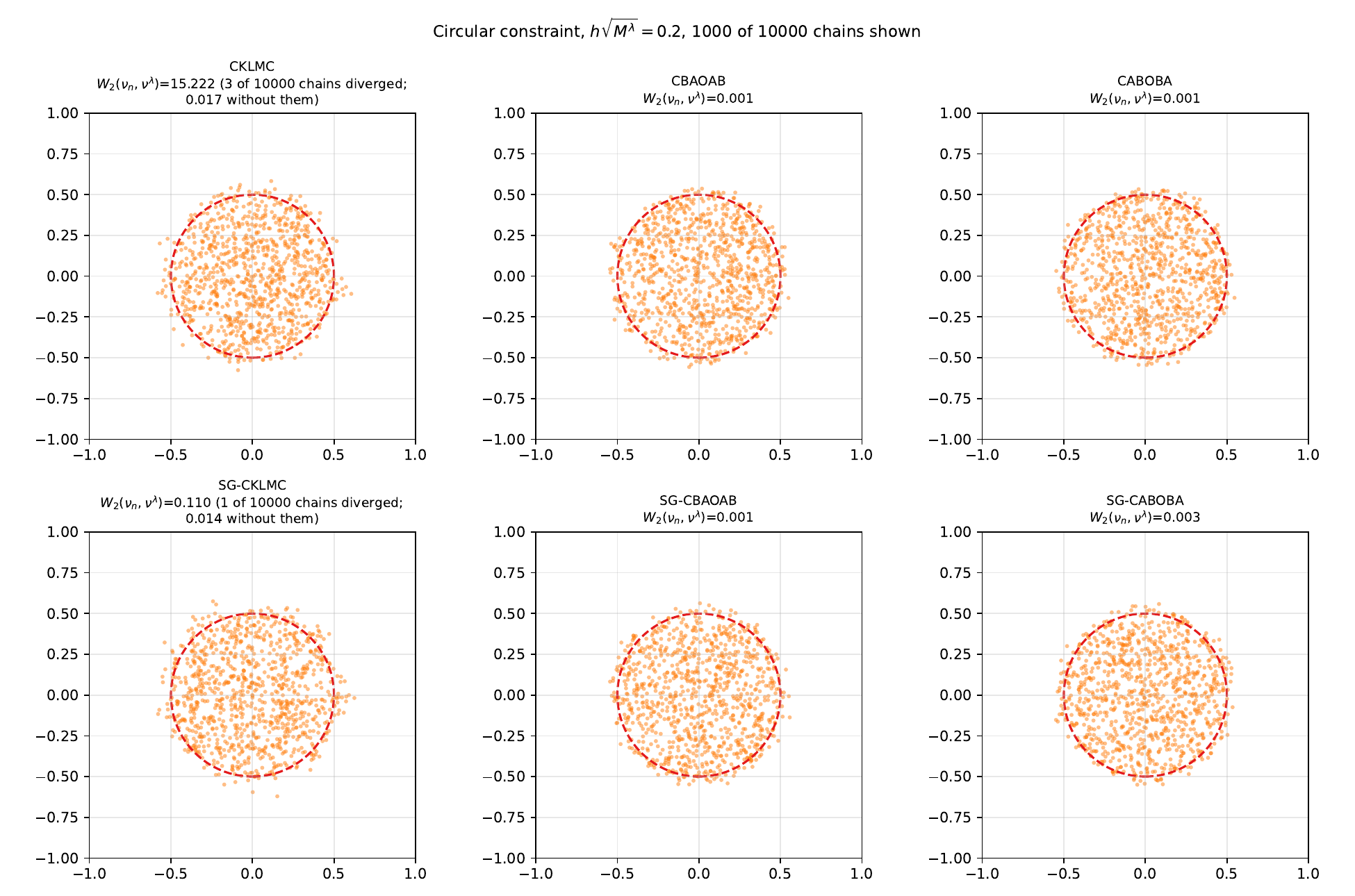}
  \caption{Circular constraint: last iterates of $1000$ chains for $h\sqrt{M^\lambda}=0.2$, with full gradients (top) and stochastic gradients (bottom). At this step size a few of the $10^4$ CKLMC chains diverge (3 with full gradients, 1 with stochastic gradients). Since $\wstwo$ averages squared distances, these chains dominate the reported error. Titles also give the error computed without them, which is still about ten times larger than for the splitting schemes.}
\label{fig:circle-scatter}
\end{figure}
 
\begin{table}[h!]
\begin{center}
\begin{tabular}{|l|c|c|c|c|}
\hline
\textbf{Algorithm} & $h\sqrt{M^\lambda}=0.1$ & $h\sqrt{M^\lambda}=0.4$ & $h\sqrt{M^\lambda}=1.6$ & ms/step \\
\hline
CKLMC & 0.0096 & unstable & unstable & 0.61 \\
CBAOAB & 0.0019 & 0.0024 & 0.0061 & 0.68 \\
CABOBA & 0.0017 & 0.0024 & 0.0095 & 0.70 \\
SG-CKLMC & 0.0087 & unstable & unstable & 1.19 \\
SG-CBAOAB & 0.0020 & 0.0020 & 0.0060 & 1.19 \\
SG-CABOBA & 0.0017 & 0.0018 & 0.0098 & 1.25 \\
\hline
\end{tabular}

\end{center}
\caption{Circular constraint: $\wstwo(\nu_n,\nu^\lambda)$ after time $T=8$ for three step sizes (mean over 3 repetitions, Monte Carlo resolution $\approx0.002$), and the cost per step in milliseconds for $10^4$ chains.}
\label{table:comp}
\end{table}
 
Figure~\ref{fig:circle-error} and Table~\ref{table:comp} show the main finding. CKLMC is only accurate at the smallest step size $h\sqrt{M^\lambda}=0.1$, where its error is about $0.01$, five times the Monte Carlo resolution. At $h\sqrt{M^\lambda}=0.2$ a fraction of its chains diverge, which increases its error by more than an order of magnitude, and for $h\sqrt{M^\lambda}\geqslant0.4$ the whole ensemble blows up. CBAOAB and CABOBA are stable for all step sizes up to $h\sqrt{M^\lambda}=1.6$, and their error is at the Monte Carlo resolution up to $h\sqrt{M^\lambda}\approx0.8$. They can therefore be run with a step size at least eight times larger than CKLMC while being more accurate, at essentially the same cost per step. The stochastic gradient versions behave like their full-gradient counterparts: the minibatch noise has no visible effect on the splitting schemes. Figure~\ref{fig:circle-scatter} illustrates the samples at $h\sqrt{M^\lambda}=0.2$.

\subsection{Exact Wasserstein errors on a box-constrained Gaussian}
\label{sec:num-box}
 
To measure errors exactly in any dimension, we consider $f(\btheta)=\sum_{i=1}^p\theta_i^2/(2\sigma_i^2)$, with the coordinates split into three equal groups with $\sigma_i\in\{0.25,0.5,1\}$ (so $m=1$ and $L=16$), $\K=[-0.3,0.6]^p$, and the Euclidean projection. The penalized gradient is separable, so all algorithms evolve the coordinates independently, and $\nu$, $\nu^\lambda$ and the law of every iterate are product measures. Hence $\wstwo$ in dimension $p$ equals $\sqrt p$ times a per-coordinate value, which we compute exactly from the one-dimensional quantile functions of $\nu$ and $\nu^\lambda$ (obtained by quadrature) and the sorted ensemble of iterates, subtracting the Monte Carlo floor. For the stochastic gradient versions, $f$ is written as a finite sum over $N_D=1000$ data points with weights $w_j>0$ of mean one, $f(\btheta)=\frac1{N_D}\sum_jw_j\sum_i\theta_i^2/(2\sigma_i^2)$, and each coordinate uses an independent minibatch, which preserves the product structure. All runs use $\gamma=2$ and $4\times10^4$ to $10^5$ independent chains.
 
 The exact approximation error decays linearly, $\wstwo(\nu^\lambda,\nu)\propto\lambda$ and $\wsone(\nu^\lambda,\nu)\propto\lambda$ (fitted slopes $1.00$ and $0.99$), which is faster than the $\lambda^{1/2+1/p}$ rate of Proposition~\ref{prop:w2_tight} for this example.
\begin{figure}[h!]
  \centering
  \includegraphics[width=\linewidth]{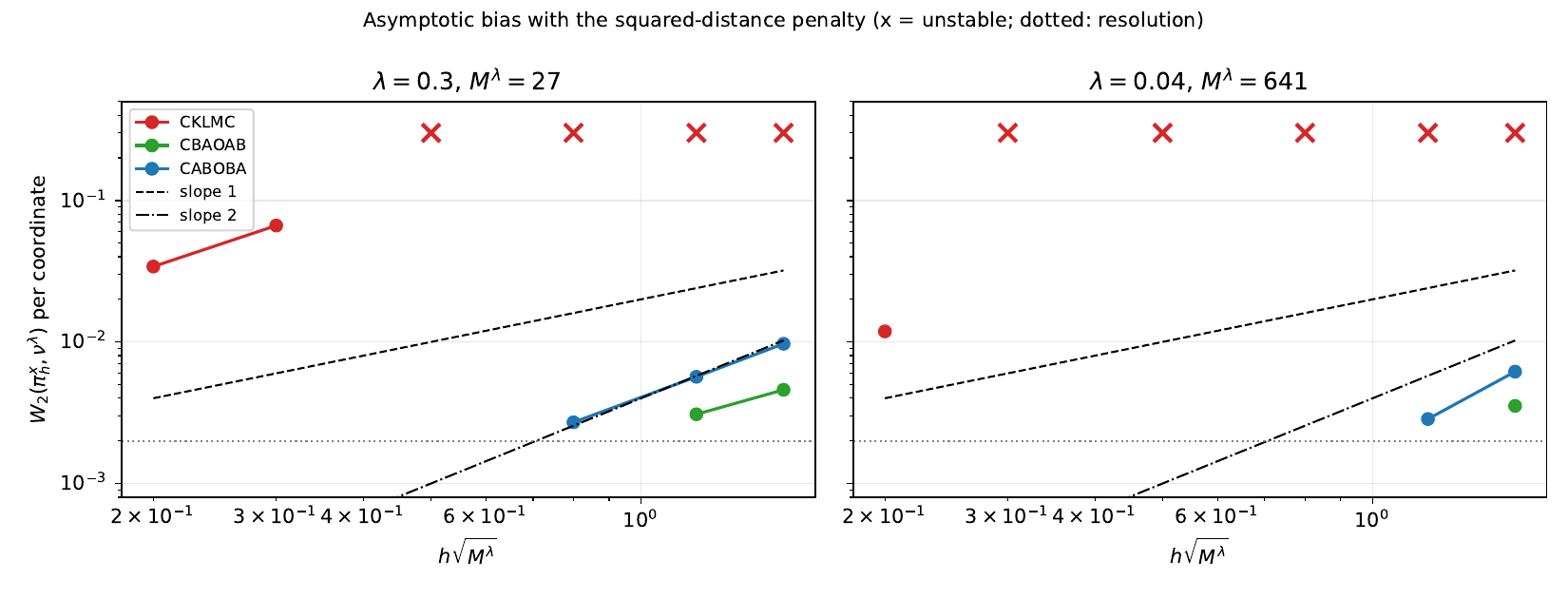}
  \caption{Box-constrained Gaussian: asymptotic bias $\wstwo(\pi_h^{\mathsf x},\nu^\lambda)$ per coordinate as a function of the step size, for $\lambda=0.3$ and $\lambda=0.04$. Values below the dotted line are below the Monte Carlo resolution and are not shown. The crosses indicate divergence.}
\label{fig:box-bias}
\end{figure}
 Figure~\ref{fig:box-bias} shows the bias of each scheme, measured after running from the exact surrogate $\nu^\lambda$ for time $6$. CKLMC diverges for $h\sqrt{M^\lambda}\geqslant0.5$ when $\lambda=0.3$ and for $h\sqrt{M^\lambda}\geqslant0.3$ when $\lambda=0.04$, and its bias at the largest stable step size is of order $10^{-2}$. The splitting schemes are stable for $h\sqrt{M^\lambda}\leqslant1.6$, and their bias is below $10^{-2}$ there and below the resolution for $h\sqrt{M^\lambda}\leqslant0.5$. For CABOBA at $\lambda=0.3$ the fitted slope is $1.8$, indicating second-order behaviour even though $U^\lambda$ does not have a Lipschitz Hessian (Remark~\ref{rem:regularity}). This is an empirical observation and is not covered by our first-order theory. CBAOAB has an even smaller bias, which is consistent with its exactness for Gaussian targets away from $\partial\K$.
 
\begin{figure}[h!]
  \centering
  \includegraphics[width=\linewidth]{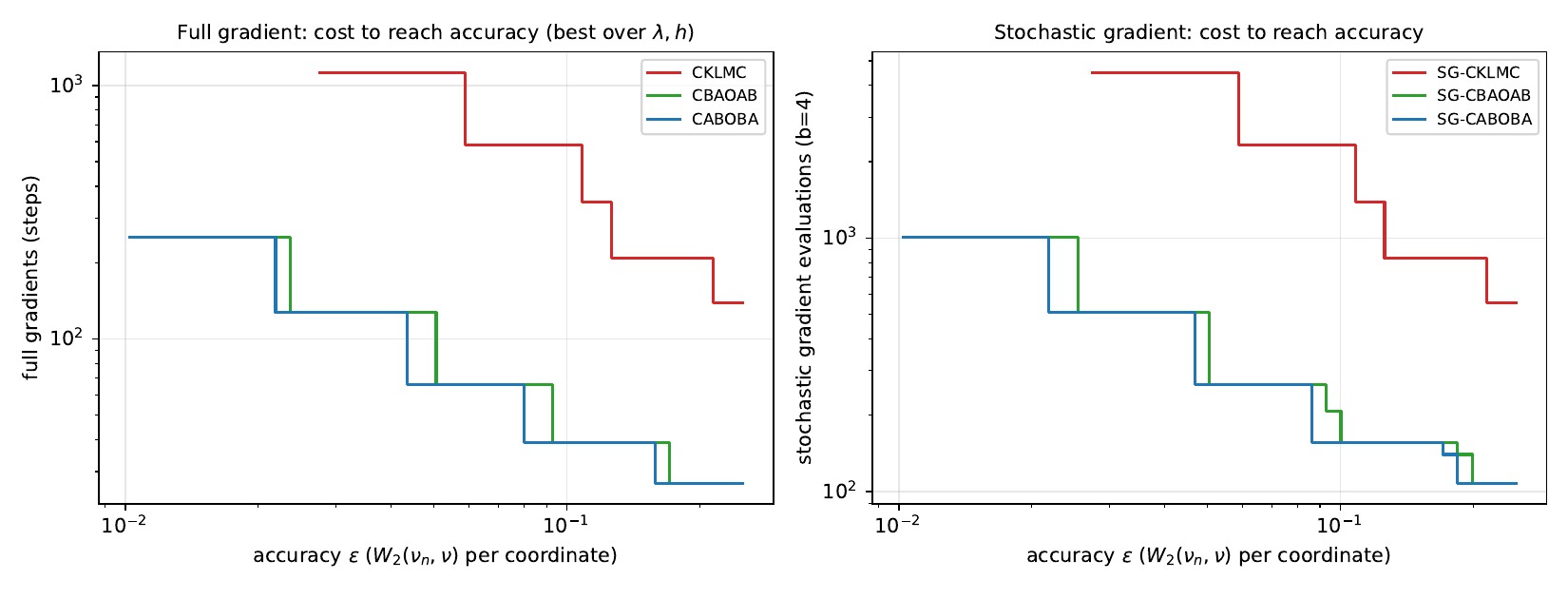}
  \caption{Box-constrained Gaussian: number of steps (left, full gradient) and of stochastic gradient evaluations (right, batch size $4$) needed to reach a given accuracy $\wstwo(\nu_n,\nu)$ per coordinate from a cold start, minimized over $\lambda\in\{0.3,0.15,0.08,0.04,0.02\}$ and $h\sqrt{M^\lambda}\in\{0.18,\dots,1.6\}$ with $T=8$.}
\label{fig:box-cost}
\end{figure}
 
We also consider the cost vs. accuracy. Figure~\ref{fig:box-cost} reports, for each accuracy level, the smallest cost over a grid of penalty parameters and step sizes, measuring the error with respect to the true constrained target $\nu$ from the cold start $\bvartheta_0=\mathbf 0$. Both splitting schemes need between $5$ and $17$ times fewer gradient evaluations than CKLMC for accuracies between $0.15$ and $0.05$ per coordinate, and within the tested step sizes CKLMC does not reach accuracies below about $0.027$, since the small values of $\lambda$ required for this make it unstable. The stochastic gradient versions with batch size $b=4$ show the same picture.
 
\begin{figure}[h!]
  \centering
  \includegraphics[width=0.55\linewidth]{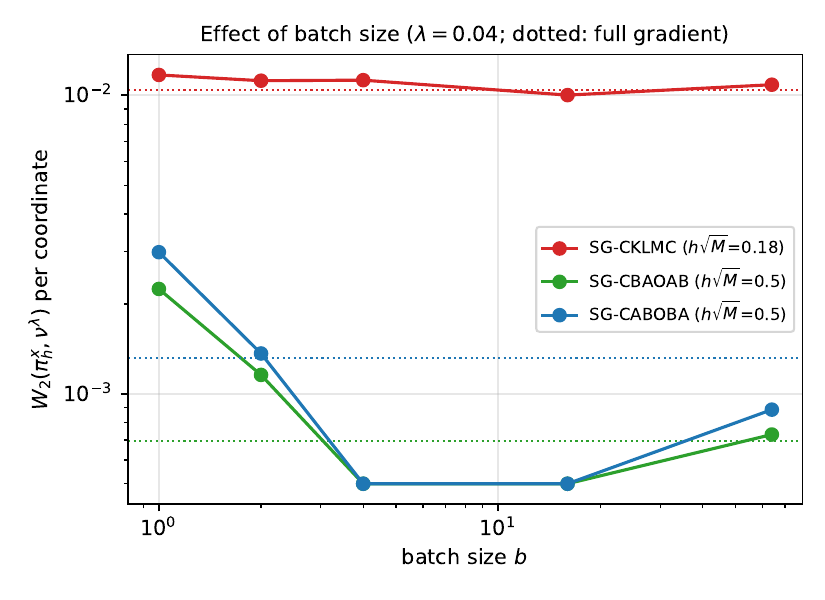}
  \caption{Box-constrained Gaussian, $\lambda=0.04$: asymptotic bias as a function of the batch size $b$ at a fixed step size for each method. The dotted lines show the full-gradient value.}
\label{fig:box-batch}
\end{figure}
 
 Finally, Figure~\ref{fig:box-batch} shows that in practice the batch size can be much smaller than required by Corollaries~\ref{cor:sg-cbaoab} and \ref{cor:sg-caboba}. For SG-CBAOAB and SG-CABOBA, even $b=1$ gives a bias of at most $0.003$ per coordinate, close to the full-gradient value and the Monte Carlo resolution. This is different to SG-CKLMC, where the bias is dominated by its discretization error for all batch sizes.

\subsection{Bayesian Constrained Linear Regression}
\label{sec:num-reg}

Our final example is a constrained Bayesian linear regression problem. If the constraint set is an $\ell_q$-ball around the origin, $q=1$ corresponds to a Bayesian Lasso-type constraint and $q=2$ to a Bayesian ridge-type constraint. We consider $q=1$ in dimension two, so that the posterior can be computed accurately on a grid. We generate $N_D=50$ data points $(a_j,y_j)$ from
\begin{equation}
\label{eq:regression}
 y_j={\btheta_{\star}}^{\top}a_j+\eta_j,\qquad a_j \sim \mathcal{N}(0,I_2),\qquad \eta_j \sim \mathcal{N}\left(0,\sigma^2\right),\qquad \sigma^2=0.25,
\end{equation}
with $\btheta_\star=(0.7,0.5)^\top$, and take a uniform prior on $\mathcal{C}=\{\btheta:\|\btheta\|_1\leq1\}$. Since $\|\btheta_\star\|_1>1$, the constraint is active. The posterior is
\begin{equation*}
\pi(\btheta) \propto \exp\big(-f(\btheta)\big)\,\mathbbm{1}_{\mathcal{C}}(\btheta),\qquad f(\btheta)=\frac1{2\sigma^2}\sum_{j=1}^{N_D}(y_j-\btheta^{\top}a_j)^2,
\end{equation*}
which has standard deviation about $0.05$ in each coordinate and concentrates near the face $\{\theta_1+\theta_2=1\}$ of $\mathcal C$. We use the Gauge penalty of $\mathcal C$, with gauge function $\|\btheta\|_1$, and $\lambda=0.01$, so that $M^\lambda=L+\lambda^{-2}\approx1.02\times10^4$, where $L\approx221$ and $m\approx194$ are the extreme eigenvalues of $\nabla^2f$. The stochastic gradients use minibatches of size $b=5$ from the $N_D=50$ data points, $\frac{N_D}{b}\sum_{j\in\omega}a_j(a_j^\top\btheta-y_j)/\sigma^2$. We use $\gamma=2\sqrt m$, run $10^4$ independent chains from $\bvartheta_0=\mathbf 0$, $\bv_0\sim\mathcal N(0,I_2)$ for time $T=30/\sqrt m\approx2.2$, with 3 repetitions, and measure the sliced Wasserstein distance to $\nu^\lambda$ as in Section~\ref{sec:num-tri}, against the exact surrogate posterior computed on a grid.

\begin{figure}[h!]
  \centering
  \includegraphics[width=\linewidth]{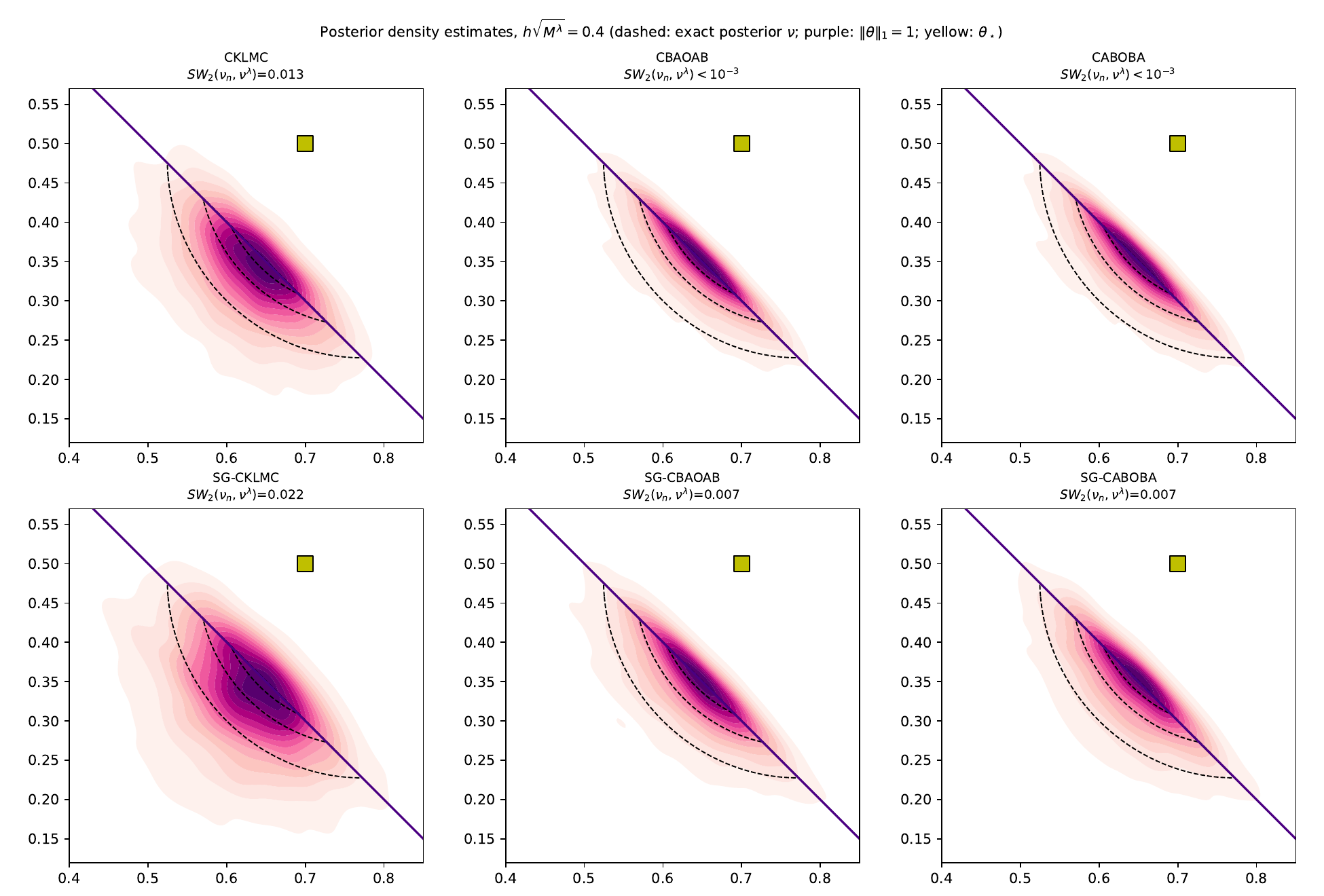}
  \caption{Constrained Bayesian linear regression, $h\sqrt{M^\lambda}=0.4$: kernel density estimates of the last iterates of the chains (filled), contours of the exact posterior (dashed), the boundary of $\mathcal C$ (purple) and $\btheta_\star$ (yellow square). Top: full gradients. Bottom: stochastic gradients with batch size $5$.}
\label{fig:reg1}
\end{figure}

\begin{figure}[h!]
  \centering
  \includegraphics[width=\linewidth]{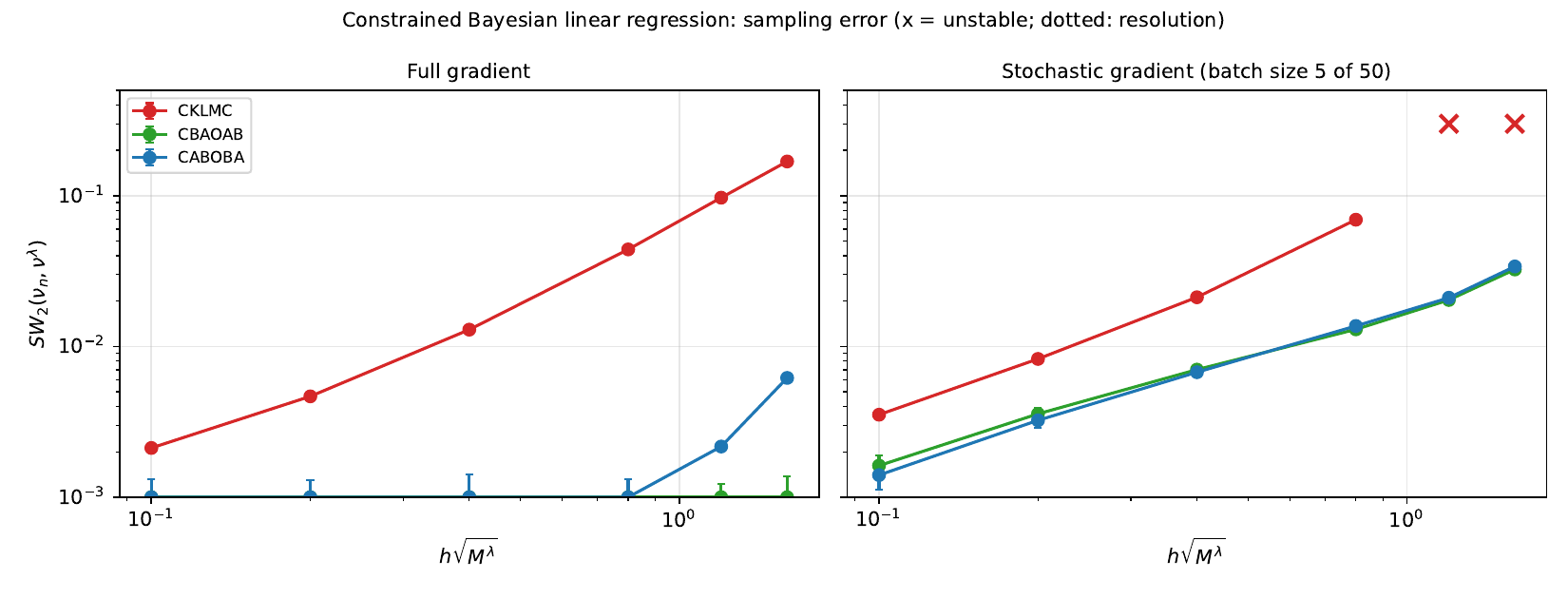}
  \caption{Constrained Bayesian linear regression: sampling error $SW_2(\nu_n,\nu^\lambda)$ after time $T$ as a function of the step size, for full gradients (left) and stochastic gradients with batch size $5$ (right). Values below the Monte Carlo resolution are shown at $10^{-3}$. The crosses indicate instability.}
\label{fig:reg2}
\end{figure}

\begin{table}[h!]
\begin{center}
\begin{tabular}{|l|c|c|c|c|}
\hline
\textbf{Algorithm} & $h\sqrt{M^\lambda}=0.1$ & $h\sqrt{M^\lambda}=0.4$ & $h\sqrt{M^\lambda}=1.2$ & ms/step \\
\hline
CKLMC & 0.0021 & 0.0129 & 0.0970 & 0.78 \\
CBAOAB & $<10^{-3}$ & $<10^{-3}$ & $<10^{-3}$ & 0.81 \\
CABOBA & $<10^{-3}$ & $<10^{-3}$ & 0.0022 & 0.79 \\
SG-CKLMC & 0.0035 & 0.0212 & unstable & 1.93 \\
SG-CBAOAB & 0.0016 & 0.0070 & 0.0203 & 1.93 \\
SG-CABOBA & 0.0014 & 0.0067 & 0.0210 & 1.98 \\
\hline
\end{tabular}

\end{center}
\caption{Constrained Bayesian linear regression: $SW_2(\nu_n,\nu^\lambda)$ for three step sizes, and the cost per step in milliseconds for $10^4$ chains.}
\label{table:reg}
\end{table}

Figures~\ref{fig:reg1}--\ref{fig:reg2} and Table~\ref{table:reg} show the results. With the larger friction $\gamma=2\sqrt m\approx28$ used here, CKLMC remains stable over the whole range of step sizes, but its error grows linearly with the step size, from $0.002$ at $h\sqrt{M^\lambda}=0.1$ to $0.1$ at $h\sqrt{M^\lambda}=1.2$, and Figure~\ref{fig:reg1} shows that its density estimate is visibly too wide. With full gradients, the errors of CBAOAB and CABOBA are at the Monte Carlo resolution for all step sizes up to $h\sqrt{M^\lambda}=0.8$, and their density estimates match the exact posterior. With stochastic gradients and the small batch size $b=5$, the minibatch noise becomes the dominant source of error for the splitting schemes and grows with the step size. Their error remains two to three times smaller than that of SG-CKLMC, which becomes unstable for $h\sqrt{M^\lambda}\geqslant1.2$. In this example the error of the stochastic gradient splitting schemes decreases with $h$, whereas the bounds of Theorems~\ref{thm:CBAOAB_SG} and \ref{thm:CABOBA_SG} contain a term which does not vanish as $h\to0$, which again indicates that these bounds are conservative in their dependence on the gradient noise.

\section{Conclusion}
\label{sec:conc}
 
The purpose of this work was to provide new algorithms for constrained sampling, based on the kinetic Langevin dynamics (KLD) and symmetric splitting schemes of the BAO family, namely CBAOAB and CABOBA, together with their stochastic gradient (SG) versions. Since the Hessian of the penalized potential is discontinuous across the boundary of the constraint set, we developed an analysis which only uses strong convexity and smoothness: one-step contraction in a shifted weighted norm, first-order bias bounds with explicit constants, and non-asymptotic Wasserstein bounds. For the stochastic gradient versions, our analysis requires batch sizes of order $\varepsilon^{-2}$ and yields complexity bounds which improve substantially on SG-CKLMC. Our experiments indicate that much smaller batches suffice in practice, which suggests that the dependence of our bounds on the gradient noise can be improved.
 
Our numerical experiments support these findings: in all examples the splitting schemes remained stable and accurate at step sizes several times larger than CKLMC, which translated into five to seventeen times fewer gradient evaluations for a given accuracy, and their stochastic gradient versions were accurate even with batch size one. The experiments also indicate second-order behaviour of the bias for the squared-distance penalty, which motivates the first direction below.
 
There are a number of directions one can take with this work, related to constrained sampling.
 
\begin{itemize}
    \item Second-order bias bounds would require more regularity of the penalized potential than the squared-distance penalty provides. Natural approaches are smoother penalties, for which a counterpart of Proposition~\ref{prop:w2_tight} would be needed, or a refined analysis exploiting that the Hessian of $U^\lambda$ is Lipschitz away from $\partial\K$.
    \item The use of additional SG-splitting schemes, which have shown promise, such as the SMS-UBU method \cite{SMSUBU}, which is based on applying minibatching without replacement.
\item Alternative strategies for constrained sampling, for example confined sampling with reflected diffusion processes, such as the work of Leimkuhler et al. \cite{Ben2016,Leimkuhler2024confined}.
\item Finally, one could apply such techniques to unbiased estimation. Recent work has analyzed this for the KLD with splitting families, including the UBUBU methodology \cite{supp:UBUBU,ruzayqat2022unbiased}. To the best of our knowledge, the only work that has done unbiased constrained sampling is that of Noble et al. \cite{noble2024unbiased}. Future ideas could also be utilized from \cite{CLL24}.
\end{itemize}
 
\section*{Acknowledgments}
NKC is supported by a City University
of Hong Kong Start-up Grant 7200809, and a Hong Kong ECS-RGC Grant 21300826. LY is supported by the City University of Hong Kong Startup Grant and Hong Kong RGC Grant 21306325. NKC is also thankful to Peter Whalley for helpful comments.
 
\bibliography{references}

\appendix
 
\section{Proofs of Theorems}
\label{sec:app_theory}
Section~\ref{app:aux} collects auxiliary results on $U^\lambda$, the shifted norm and two general perturbation lemmas. Section~\ref{app:contraction} proves Proposition~\ref{prop:contraction}. Section~\ref{app:local} contains the local error estimates of CBAOAB and CABOBA with respect to an exact measure-preserving scheme, from which we derive the bias bounds. Section~\ref{app:main-proofs} proves the theorems and corollaries for the splitting schemes, and Section~\ref{app:sgcklmc} contains the proofs for SG-CKLMC.
 Throughout we write $U=U^\lambda$, $M=M^\lambda$, $G:=\nabla U^\lambda$ and $s:=h/2$. Recall that \eqref{eq:stepsize} implies $h\sqrt M<1/2$, hence
\begin{equation}
\label{eq:s-small}
s\sqrt M<\tfrac14 .
\end{equation}

\subsection{Auxiliary results}
\label{app:aux}

Since $\nabla U$ is $M$-Lipschitz, by Rademacher's theorem  and the mean value theorem for $C^{1,1}$ functions
\begin{equation}
\label{eq:ae-hessian}
mI_p\preccurlyeq\nabla^2U(x)\preccurlyeq MI_p\ \text{ for a.e. }x,
\end{equation}
\begin{equation}
\label{eq:mvt}
\nabla U(x)-\nabla U(x')=Q(x-x')\ \text{ for some symmetric $Q$ with } mI_p\preccurlyeq Q\preccurlyeq MI_p .
\end{equation}

\begin{lemma}
\label{lem:Ulam}
Under Assumptions~\ref{asm:radius}--\ref{asm:smooth}:
(i) $\nabla U^\lambda(\mathbf 0)=\mathbf 0$, and $\mathbf 0$ is the unique minimizer of $U^\lambda$.
(ii) $\E_{\nu^\lambda}\|\bvartheta\|^2\leqslant p/m$.
(iii) $\E_{\nu^\lambda}\|\nabla U^\lambda(\bvartheta)\|^2\leqslant Mp$.
(iv) the Hamiltonian flow $\varphi_t$ of $\dot x=v,\ \dot v=-\nabla U^\lambda(x)$ and the maps $\mathcal O_t$ preserve $\pi^\lambda$.
\end{lemma}
\begin{proof}
(i) $\nabla f(\mathbf 0)=\mathbf 0$ by Assumption~\ref{asm:f}, and $\nabla\ell^\lambda_\K(\mathbf 0)=\mathbf 0$ since $\mathbf 0\in{\rm int}\,\K$ and $d_\K\equiv0$ on $\K$, uniqueness follows from strong convexity.
(ii)--(iii) Integration by parts against $e^{-U}$ (boundary terms vanish since $U(x)\geqslant U(\mathbf 0)+\frac m2\|x\|^2$ and $\|\nabla U(x)\|\leqslant M\|x\|$) gives
\[
\E_{\nu^\lambda}\langle\nabla U(\bvartheta),\bvartheta\rangle=p,\qquad \E_{\nu^\lambda}\|\nabla U(\bvartheta)\|^2=\E_{\nu^\lambda}{\rm tr}\,\nabla^2U(\bvartheta).
\]
Conclude with $\langle\nabla U(x),x\rangle\geqslant m\|x\|^2$, by (i), and ${\rm tr}\,\nabla^2U\leqslant Mp$, by \eqref{eq:ae-hessian}.
(iv) $\mathcal O_t$ fixes $x$ and preserves $\mathcal N(0,I_p)$. The field $(v,-\nabla U(x))$ is Lipschitz and divergence-free, so $\varphi_t$ preserves Lebesgue measure, and $\frac{d}{dt}\big(U(x_t)+\frac12\|v_t\|^2\big)=0$, hence $\varphi_t$ preserves $\pi^\lambda\propto e^{-U(x)-\|v\|^2/2}$.
\end{proof}

\begin{lemma}
\label{lem:norms}
Assume \eqref{eq:stepsize}, and let $a,b$ be as in \eqref{eq:abc}. For all $y,w\in\R^p$,
\begin{itemize}
\item[(i)] $\frac12(\|y\|^2+a\|w\|^2)\leqslant\|(y,w)\|_{a,b}^2\leqslant\frac32(\|y\|^2+a\|w\|^2)$.
\item[(ii)] $\nsharp{(y,w)}\leqslant\sqrt{3/2}\,\big(\|y\|+\tfrac54\sqrt a\,\|w\|\big)$, and $\nsharp{(y,\mathbf 0)}=\|y\|$.
\item[(iii)] $\|y\|\leqslant\sqrt{17/8}\,\nsharp{(y,w)}\leqslant\frac32\nsharp{(y,w)}$ and $\|y+sw\|\leqslant\sqrt{5/2}\,\nsharp{(y,w)}\leqslant\frac85\nsharp{(y,w)}$.
\end{itemize}
Consequently $\wstwo(\mu^{\mathsf x},\mu'^{\mathsf x})\leqslant\frac32\wsharp(\mu,\mu')$ for $\mu,\mu'\in\mathcal P_2(\R^{2p})$.
\end{lemma}
\begin{proof}
By \eqref{eq:stepsize} and \eqref{eq:s-small}, $b<\frac12\sqrt a$ and $s\leqslant\frac14\sqrt a$. (i) follows from $2b|\langle y,w\rangle|\leqslant\frac b{\sqrt a}(\|y\|^2+a\|w\|^2)$. For (ii) and (iii), apply (i) to $T^{-1}(y,w)=(y-sw,w)$ and write $r^2:=\|y-sw\|^2+a\|w\|^2\leqslant2\nsharp{(y,w)}^2$, then by Cauchy--Schwarz, we have
\begin{align*}
\nsharp{(y,w)}&\leqslant\sqrt{\tfrac32}\big(\|y\|+(s+\sqrt a)\|w\|\big), 
\\
\|y\|&\leqslant\|y-sw\|+\tfrac{\sqrt a}4\|w\|\leqslant\sqrt{\tfrac{17}{16}}\,r, \\
\|y+sw\|&\leqslant\|y-sw\|+\tfrac{\sqrt a}2\|w\|\leqslant\sqrt{\tfrac54}\,r,
\end{align*}
The final statement follows by applying (iii) to an optimal coupling.
\end{proof}

Let $P(z,\cdot)=\Law(\Phi(z,\bxi))$, $\bxi\sim\varrho$, be a Markov kernel on $\R^{2p}$ and $\|\cdot\|_\star$ a norm induced by an inner product such that, for some $c\in(0,1)$,
\begin{equation}
\label{eq:gen-contraction}
\|\Phi(z,\bxi)-\Phi(z',\bxi)\|_\star\leqslant(1-c)^{1/2}\|z-z'\|_\star\qquad\text{for all }z,z',\bxi .
\end{equation}
The next result is a lemma based on the bias of the local error.
\begin{lemma}
\label{lem:bias-generic}
Assume \eqref{eq:gen-contraction}. Then $P$ has a unique invariant measure $\pi_P\in\mathcal P_2(\R^{2p})$. If $\Psi(z,\bxi)$ is a random map whose kernel $S$ leaves $\mu\in\mathcal P_2(\R^{2p})$ invariant, and $\delta:=\big(\E\|\Phi(Z,\bxi)-\Psi(Z,\bxi)\|_\star^2\big)^{1/2}$ with $Z\sim\mu$ independent of $\bxi$, then ${\sf W}_{2,\star}(\pi_P,\mu)\leqslant\delta/(1-\sqrt{1-c})\leqslant2\delta/c$.
\end{lemma}
\begin{proof}
Synchronous coupling and \eqref{eq:gen-contraction} give
\begin{equation}
\label{eq:W-contraction}
{\sf W}_{2,\star}(\mu_1P,\mu_2P)\leqslant(1-c)^{1/2}\,{\sf W}_{2,\star}(\mu_1,\mu_2),
\end{equation}
so $\pi_P$ exists and is unique by the Banach fixed point theorem, and $\mu P^n\to\pi_P$. Since ${\sf W}_{2,\star}(\mu P,\mu)={\sf W}_{2,\star}(\mu P,\mu S)\leqslant\delta$, \eqref{eq:W-contraction} gives ${\sf W}_{2,\star}(\mu P^n,\mu)\leqslant\sum_{k\geqslant0}(1-c)^{k/2}\delta$ and let $n\to\infty$.
\end{proof}

\subsection{Proof of Proposition~\ref{prop:contraction}}
\label{app:contraction}

Let $T:=\mathcal A_{s}$, $\eta:=e^{-\gamma h}$, and $\|z\|_{a,b}^2=z^\top\mathbf Gz$ with $\mathbf G:=\begin{pmatrix}I_p&bI_p\\ bI_p&aI_p\end{pmatrix}$. Set $X_1:=\mathcal A_h\mathcal B_s\mathcal O_h\mathcal B_s$ and $X_2:=\mathcal A_s\mathcal B_h\mathcal A_s\mathcal O_h$ (in the order of application). Then $\mathcal A_{-s}X_1\mathcal A_s$ is the CABOBA step and $\mathcal A_{-s}X_2\mathcal A_s$ is the CBAOAB interior step, so $\Phi(z,\bxi)=T(X_i(T^{-1}z,\bxi))$ and, $T$ being linear,
\begin{equation}
\label{eq:conj-norm}
\nsharp{\Phi(z,\bxi)-\Phi(z',\bxi)}=\|X_i(T^{-1}z,\bxi)-X_i(T^{-1}z',\bxi)\|_{a,b}.
\end{equation}
Each $X_i$ evaluates $\nabla U$ at one point per chain, so by \eqref{eq:mvt} the difference of two synchronously coupled steps is linear: $X_i(y,\bxi)-X_i(y',\bxi)=\mathbf P_i(Q)(y-y')$ with $mI_p\preccurlyeq Q\preccurlyeq MI_p$ and
\begin{equation}
\label{eq:diff-maps}
\mathbf P_1(Q)=\begin{pmatrix}I_p& hI_p\\ -\frac h2(1+\eta)Q& \eta I_p-\frac{h^2}2(1+\eta)Q\end{pmatrix},\qquad
\mathbf P_2(Q)=\begin{pmatrix}I_p-\frac{h^2}2Q& hI_p-\frac{h^3}4Q\\ -\eta hQ& \eta\big(I_p-\frac{h^2}2Q\big)\end{pmatrix}.
\end{equation}
These are the difference maps of the steps ``ABOB'' and ``ABAO'' in the proofs of \cite[Theorems 5.2 and 5.1]{leimkuhler2023contractiona}, where it is shown that, under \eqref{eq:stepsize},
\begin{equation}
\label{eq:matrix-contraction}
\mathbf P_i(Q)^\top\mathbf G\,\mathbf P_i(Q)\preccurlyeq(1-c(h))\,\mathbf G\qquad\text{for all symmetric $Q$ with } mI_p\preccurlyeq Q\preccurlyeq MI_p
\end{equation}
(a $2\times2$ inequality for each eigenvalue of $Q$, see also \cite{monmarche2021high}). By \eqref{eq:conj-norm}--\eqref{eq:matrix-contraction}, \eqref{eq:gen-contraction} holds with $\|\cdot\|_\star=\nsharp{\cdot}$ and $c=c(h)$, and Lemma~\ref{lem:bias-generic} gives the Wasserstein contraction and $\pi_h$. \qed

\subsection{Local error estimates and bias}
\label{app:local}

We compare both schemes with the HOH scheme $\varphi_s\,\mathcal O_h\,\varphi_s$ of \cite[Section 8]{leimkuhler2023contractiona}, which preserves $\pi^\lambda$ by Lemma~\ref{lem:Ulam}(d). Let $(x,v)\sim\pi^\lambda$, let $\bxi$ be the common Gaussian increment, and set
\[
(x(r),v(r)):=\varphi_r(x,v),\quad (x_1,v_1):=(x(s),v(s)),\quad v_O:=\etah v_1+\sqrt{1-\etah^2}\,\bxi,\quad (\bar x(r),\bar v(r)):=\varphi_r(x_1,v_O),
\]
so that $(x_2,v_2):=(\bar x(s),\bar v(s))$ is the HOH step. All these pairs have law $\pi^\lambda$, so by Lemma~\ref{lem:Ulam}(iii), for $r\in[0,s]$,
\begin{equation}
\label{eq:stat-moments}
\|v(r)\|_{\Ltwo}=\|\bar v(r)\|_{\Ltwo}=\|v_O\|_{\Ltwo}=\sqrt p,\qquad \|G(x(r))\|_{\Ltwo},\ \|G(\bar x(r))\|_{\Ltwo}\leqslant\sqrt{Mp}.
\end{equation}
Together with $x(r)=x+rv-\int_0^r(r-u)G(x(u))du$, $v(r)=v-\int_0^rG(x(u))du$ (and similarly for $\bar x,\bar v$), Minkowski's inequality and the $M$-Lipschitz continuity of $G$, this gives all bounds below.

\begin{lemma}
\label{lem:local-aboba}
Assume \eqref{eq:stepsize}, and let $(x^c,v^c)$ be the CABOBA step from $(x,v)\sim\pi^\lambda$ with increment $\bxi$. Then
\[
\|x^c-x_2\|_{\Ltwo}\leqslant\tfrac98Ms^3\sqrt p,\qquad \|v^c-v_2\|_{\Ltwo}\leqslant\tfrac{29}{24}Ms^2\sqrt p,\qquad
\big(\E\nsharp{(x^c-x_2,v^c-v_2)}^2\big)^{1/2}\leqslant\tfrac35h^2\sqrt{Mp}.
\]
\end{lemma}
\begin{proof}
With $\tilde x=x+sv$, the step is $v^b=\etah(v-sG(\tilde x))+\sqrt{1-\etah^2}\,\bxi$, $v^c=v^b-sG(\tilde x)$, $x^c=\tilde x+sv^c$. Writing 
$$
\int_0^s(s-r)G(z(r))dr=\int_0^s(s-r)(G(z(r))-G(\tilde x))dr+\frac{s^2}2G(\tilde x)
$$for $z\in\{x,\bar x\}$, the terms in $G(\tilde x)$ cancel and
\begin{align*}
x^c-x_2&=\int_0^s(s-r)\big(G(x(r))-G(\tilde x)\big)dr+\int_0^s(s-r)\big(G(\bar x(r))-G(\tilde x)\big)dr+s(v^b-v_O),\\
v^c-v_2&=(v^b-v_O)+\int_0^s\big(G(\bar x(r))-G(\tilde x)\big)dr,\qquad v^b-v_O=\etah\int_0^s\big(G(x(r))-G(\tilde x)\big)dr .
\end{align*}
Inserting $\|x(r)-\tilde x\|_{\Ltwo}\leqslant(s-r)\sqrt p+\frac{r^2}2\sqrt{Mp}$ and $\|\bar x(r)-\tilde x\|_{\Ltwo}\leqslant r\sqrt p+\frac{s^2+r^2}2\sqrt{Mp}$ and integrating,
\[
\|x^c-x_2\|_{\Ltwo}\leqslant Ms^3\sqrt p+\tfrac12M^{3/2}s^4\sqrt p,\qquad \|v^c-v_2\|_{\Ltwo}\leqslant Ms^2\sqrt p+\tfrac56M^{3/2}s^3\sqrt p,
\]
and \eqref{eq:s-small} gives the first two claims. For the third, Lemma~\ref{lem:norms}(ii) bounds the left-hand side by $\sqrt{3/2}\,\big(\|x^c-x_2\|_{\Ltwo}+\frac5{4\sqrt M}\|v^c-v_2\|_{\Ltwo}\big)\leqslant\sqrt{3/2}\,\big(\frac9{128}+\frac{145}{384}\big)h^2\sqrt{Mp}\leqslant0.55\,h^2\sqrt{Mp}$.
\end{proof}

\begin{lemma}
\label{lem:local-baoab}
Assume \eqref{eq:stepsize}, and let $(x^c,v^c)$ be the CBAOAB step from $(x,v)\sim\pi^\lambda$ with increment $\bxi$. Then $\|x^c-x_2\|_{\Ltwo}\leqslant\frac{35}{24}Ms^3\sqrt p$, $\|v^c-v_2\|_{\Ltwo}\leqslant\frac{63}{50}Ms^2\sqrt p$, and, with $\mathcal B_{-s}(y,w):=(y,w+sG(y))$,
$\big(\E\nsharp{\mathcal B_{-s}(x^c,v^c)-\mathcal B_{-s}(x_2,v_2)}^2\big)^{1/2}\leqslant\frac23h^2\sqrt{Mp}$.
\end{lemma}
\begin{proof}
The step is $v^a=v-sG(x)$, $v^b=\etah v^a+\sqrt{1-\etah^2}\,\bxi$, $x^c=x+sv^a+sv^b$, $v^c=v^b-sG(x^c)$. With
$J_1:=\int_0^s(s-r)(G(x(r))-G(x))dr$, $J_2:=\int_0^s(G(x(r))-G(x))dr$, $J_3:=\int_0^s(s-r)(G(\bar x(r))-G(x))dr$ and $J_4:=\int_0^s(G(\bar x(r))-G(x^c))dr$, the same cancellation as before gives
\[
x^c-x_2=J_1+s\etah J_2+J_3,\qquad v^c-v_2=\etah J_2+J_4 .
\]
Using $\|x(r)-x\|_{\Ltwo}\leqslant r\sqrt p+\frac{r^2}2\sqrt{Mp}$, $\|\bar x(r)-x\|_{\Ltwo}\leqslant(s+r)\sqrt p+\frac{s^2+r^2}2\sqrt{Mp}$ and $\|x^c-\bar x(r)\|_{\Ltwo}\leqslant\|x^c-x_2\|_{\Ltwo}+(s-r)\sqrt p$, and
\[
\|J_1\|_{\Ltwo}\leqslant M\Big(\tfrac{s^3}6\sqrt p+\tfrac{s^4}{24}\sqrt{Mp}\Big),\quad
\|J_2\|_{\Ltwo}\leqslant M\Big(\tfrac{s^2}2\sqrt p+\tfrac{s^3}6\sqrt{Mp}\Big),\quad
\|J_3\|_{\Ltwo}\leqslant M\Big(\tfrac{2s^3}3\sqrt p+\tfrac{7s^4}{24}\sqrt{Mp}\Big),
\]
we obtain
\[
\|x^c-x_2\|_{\Ltwo}\leqslant\tfrac43Ms^3\sqrt p+\tfrac12M^{3/2}s^4\sqrt p,\qquad
\|v^c-v_2\|_{\Ltwo}\leqslant Ms^2\sqrt p+\tfrac16M^{3/2}s^3\sqrt p+Ms\|x^c-x_2\|_{\Ltwo},
\]
and \eqref{eq:s-small} gives the first two claims. The velocity component of $\mathcal B_{-s}(x^c,v^c)-\mathcal B_{-s}(x_2,v_2)$ is $v^c-v_2+s(G(x^c)-G(x_2))$, of $\Ltwo$ norm at most $(\frac{63}{50}+\frac{35}{384})Ms^2\sqrt p$, and Lemma~\ref{lem:norms}(b) gives the third claim, since $\sqrt{3/2}\,\big(\frac{35}{96}+\frac54(\frac{63}{50}+\frac{35}{384})\big)s^2\leqslant0.63\,h^2$.
\end{proof}

\begin{proposition}
\label{prop:bias-cbaoab}
Assume \eqref{eq:stepsize}. Let $\pi_h^{\rm ABOBA}$ and $\pi_h^{\rm int}$ be the invariant measures of CABOBA and of the CBAOAB interior chain, and $\mu^{\rm int}:=\Law(\mathcal B_{-s}(Z))$ for $Z\sim\pi^\lambda$. Then
\[
\wsharp(\pi_h^{\rm ABOBA},\pi^\lambda)\leqslant\tfrac{24}5\mathsf B_h,\qquad \wsharp(\pi_h^{\rm int},\mu^{\rm int})\leqslant\tfrac{16}{3}\mathsf B_h,
\]
and $\wstwo((\pi_h^{\rm ABOBA})^{\mathsf x},\nu^\lambda)\leqslant\frac{36}5\mathsf B_h$, $\wstwo((\pi_h^{\rm int})^{\mathsf x},\nu^\lambda)\leqslant8\mathsf B_h$. The position marginal of $\pi_h^{\rm int}$ is that of the invariant measure of CBAOAB.
\end{proposition}
\begin{proof}
Apply Lemma~\ref{lem:bias-generic} with $\|\cdot\|_\star=\nsharp{\cdot}$ and $c=c(h)$. For CABOBA with $\Psi$ the HOH step and $\mu=\pi^\lambda$, where $\delta\leqslant\frac35h^2\sqrt{Mp}$ by Lemma~\ref{lem:local-aboba}. For the interior chain $\Phi=\mathcal B_{-s}\circ\Phi_{\rm BAOAB}\circ\mathcal B_s$ with $\Psi=\mathcal B_{-s}\circ\Psi_{\rm HOH}\circ\mathcal B_s$, which leaves $\mu^{\rm int}$ invariant, where $\delta\leqslant\frac23h^2\sqrt{Mp}$ by Lemma~\ref{lem:local-baoab}. Since $2\delta/c(h)=8(1-e^{-\gamma h})\delta/(h^2m)$, this gives the two $\wsharp$ bounds. The position marginals of $\pi^\lambda$ and $\mu^{\rm int}$ are $\nu^\lambda$, so Lemma~\ref{lem:norms} gives the position bounds. Finally $\mathcal B_{-s}$ does not change positions.
\end{proof}

\subsection{Proofs of the main results for the splitting schemes}
\label{app:main-proofs}
 
% \begin{lemma}
% \label{lem:sg-increments}
% Let $e:=\mathcal G(y,\omega)-\nabla f(y)$ denote the error of the stochastic gradient at the point $y$ where it is evaluated. Then:
% \begin{itemize}
% \item[(i)] One step of SG-CABOBA from $z$ equals the CABOBA step $\Phi(z,\bxi)$ plus $\Delta=(-\frac{h^2}4(1+\etah)e,\,-\frac h2(1+\etah)e)$ with $y=\tilde x=x+\frac h2v$, and $\nsharp{\Delta}^2\leqslant\frac{h^2}{M}\|e\|^2$.
% \item[(ii)] One step of the SG-CBAOAB interior chain from $z=(x,w)$ equals the interior step $\Phi(z,\bxi)$ plus $\Delta=(-\frac h2(1+\etah)he,\,-\etah he)$ with $y=x$, and $\nsharp{\Delta}^2\leqslant\frac{8h^2}{5M}\|e\|^2$.
% \end{itemize}
% \end{lemma}
% \begin{proof}
% (i) The stochastic gradient replaces $G(\tilde x)$ by $G(\tilde x)+e$ in both half kicks, the velocity after the second kick changes by $-\frac h2(1+\etah)e$ and the final position by $\frac h2$ times this. Hence $T^{-1}\Delta=(\mathbf 0,-\frac h2(1+\etah)e)$ and $\nsharp{\Delta}^2=a\frac{h^2}4(1+\etah)^2\|e\|^2\leqslant\frac{h^2}M\|e\|^2$.
% (b) The kick $\mathcal B_h$ adds $-he$ to the velocity, $\mathcal A_{h/2}$ adds $-\frac h2he$ to the position, $\mathcal O_h$ multiplies the velocity perturbation by $\etah$, $\mathcal A_{h/2}$ adds $-\frac h2\etah he$ to the position. Hence $T^{-1}\Delta=(-\frac h2he,-\etah he)$, and Lemma~\ref{lem:norms}(i) gives $\nsharp{\Delta}^2\leqslant\frac32h^2(\frac{h^2}4+\frac1M)\|e\|^2\leqslant\frac32\cdot\frac{17}{16}\frac{h^2}{M}\|e\|^2$.
% \end{proof}
 
\begin{proof}[Proof of Theorem~\ref{thm:cbaoab}]
Let $\bz_k$ be the interior chain \eqref{eq:interior}, whose position is $\bvartheta_k$. With $\bvartheta_0=\mathbf 0$ and Lemma~\ref{lem:Ulam}(i), $\bz_0=(\mathbf 0,\bv_0)$. By the triangle inequality and Lemma~\ref{lem:norms},
\[
\wstwo(\nu_n,\nu)\leqslant\tfrac32\wsharp(\Law(\bz_n),\pi_h^{\rm int})+\wstwo((\pi_h^{\rm int})^{\mathsf x},\nu^\lambda)+\wstwo(\nu^\lambda,\nu).
\]
By Proposition~\ref{prop:contraction}, $\wsharp(\Law(\bz_n),\pi_h^{\rm int})\leqslant(1-c(h))^{n/2}\big(\wsharp(\Law(\bz_0),\mu^{\rm int})+\wsharp(\mu^{\rm int},\pi_h^{\rm int})\big)$. Coupling $\bz_0=(\mathbf 0,\bv)$ with $(x,\bv+sG(x))\sim\mu^{\rm int}$, $x\sim\nu^\lambda$, Lemma~\ref{lem:norms}(ii) and Lemma~\ref{lem:Ulam} gives
$$
\wsharp(\Law(\bz_0),\mu^{\rm int})\leqslant\sqrt{3/2}\big(\|x\|_{\Ltwo}+s^2\|G(x)\|_{\Ltwo}+\sqrt a\,s\|G(x)\|_{\Ltwo}\big)\leqslant\sqrt{3/2}\,\frac{21}{16}\sqrt{p/m}\leqslant\frac{13}{8}\sqrt{p/m}.
$$Combining with Proposition~\ref{prop:bias-cbaoab} and Proposition~\ref{prop:w2_tight} gives the $\wstwo$ bound.
% since 
% $\frac32\cdot\frac{13}8\leqslant\frac52$ and $\frac32\cdot\frac{16}3=8$. 
For $\wsone$ we use $\wsone(\nu_n,\nu)\leqslant\wstwo(\nu_n,\nu^\lambda)+\wsone(\nu^\lambda,\nu)$. Finally $1-e^{-\gamma h}\leqslant\gamma h$ gives $c(h)\geqslant mh/(4\gamma)$, leading to the final results%  and $(1-c(h))^{n/2}\leqslant e^{-nc(h)/2}$.
\end{proof}
 
\begin{proof}[Proof of Theorem~\ref{thm:caboba}]
The proof is identical to the above proof, using $\Phi_{\rm ABOBA}$ and $\pi_h^{\rm ABOBA}$ in place of the interior chain, and $\wsharp(\delta_{\mathbf 0}\otimes\mathcal N(0,I_p),\pi^\lambda)\leqslant\sqrt{p/m}$ by Lemma~\ref{lem:norms}(ii).
\end{proof}
 
\subsubsection*{Proofs of convergence for SG-CBAOAB}
 
We first introduce the coupling used to compare one stochastic-gradient CBAOAB step with its full-gradient counterpart. For notational simplicity, throughout this section we omit the superscript $\mathsf{SG\text{-}CBAOAB}$ whenever no ambiguity arises. We also use $\mathsf{sg}$ to denote the stochastic gradient scheme, and $\mathsf{fg}$ to denote the full gradient scheme. \textcolor{black}{Since Algorithm~\ref{alg:SG-CBAOAB} reuses the stochastic gradient computed at $\bvartheta_k$ in the two half kicks surrounding it, the pair $(\bvartheta_k,\bv_k)$ is not by itself a Markov chain. We therefore work with the interior chain $Z_k=(\bvartheta_k,\bv_k+\frac h2\widetilde\nabla U^\lambda(\bvartheta_k,\omega_{k+1}))$, which is a Markov chain whose position is $\bvartheta_k$, and whose step is $\mathcal B_h\mathcal A_{h/2}\mathcal O_h\mathcal A_{h/2}$, where $\mathcal B_h$ uses the stochastic gradient at the current position (see \eqref{eq:interior}). All distances below are measured in the shifted norm $\nsharp{\cdot}$ of \eqref{eq:sharp}, in which the full-gradient kernel contracts at every step (Proposition~\ref{prop:contraction}).}
 
Fix $k \geqslant 0$. Let
$
Z_k^{\sf sg}=(\bvartheta_k^{\sf sg},\bv_k^{\sf sg})
$
denote the current state of the stochastic-gradient \textcolor{black}{interior} chain, and define the field
$
\mathcal F_k:=\sigma(\bvartheta_k^{\sf sg},\bv_k^{\sf sg}).
$
Starting from the same initial state \(Z_k^{\sf sg}\), we couple one step using the full gradient with one step using the stochastic gradient, with the same Gaussian noise used in the O-step.
Denote by $P_h$ and $Q_h$ the \textcolor{black}{interior} kernel with step size $h$ using the full gradient and the stochastic gradient, respectively.
Let
\[
Y_{k+1}\sim P_h(Z_k^{\sf sg},\cdot),
\]
be the next state produced by one full-gradient step, and let
\[
Z_{k+1}^{\sf sg}\sim Q_h(Z_k^{\sf sg},\cdot),
\]
be the next state produced by one stochastic-gradient step. We then define the one-step discrepancy by
\[
\Delta Z_{k+1}:=Z_{k+1}^{\sf sg}-Y_{k+1}
=(\Delta \bvartheta_{k+1},\Delta \bv_{k+1}).
\]
 
\begin{lemma}
\label{lem:sg-cbaoab}
When \eqref{eq:stepsize} holds, it holds that
\begin{align*}
\E[\nsharp{\Delta Z_{k+1}}^2|\mathcal{F}_k]
\leqslant \textcolor{black}{C_V\,h^2\,\|\bvartheta_k^{\sf sg}\|^2,\qquad C_V:=\frac85\,\frac{\sigma_1^2L^2}{M}.}
\end{align*}
\end{lemma}
\begin{proof}
Given the $k$-th state $(\bvartheta_k,\bv_k)$ of the \textcolor{black}{interior} chain, we perform one update from this state using full gradients $\nabla U^\lambda$ and also the stochastic gradients $\tilde \nabla U^\lambda$.
		\begin{itemize}
			\item[(B)] $\bv_k^{\sf (1), fg} = \bv_{k} - h\nabla U^\lambda(\bvartheta_k)$ and $\bv_k^{\sf (1), sg} = \bv_{k} - h\tilde \nabla U^\lambda(\bvartheta_k)$
                \item[(A)] $\bvartheta_k^{\sf (1), fg} = \bvartheta_{k} + \frac{h}{2}\bv_k^{\sf (1), fg}$ and $\bvartheta_k^{\sf (1), sg} = \bvartheta_{k} + \frac{h}{2}\bv_k^{\sf (1), sg}$
                \item[] Sample $\xi_{k} \sim \mathcal{N}(0_{p},I_{p})$
                \item[(O)] $\bv_k^{\sf (2),fg} = e^{-\gamma h} \bv_k^{\sf (1),fg} + \sqrt{1-e^{-2\gamma h}}\,\xi_{k}$ and $\bv_k^{\sf (2),sg} = e^{-\gamma h} \bv_k^{\sf (1),sg} + \sqrt{1-e^{-2\gamma h}}\,\xi_{k}$
                \item[(A)] $\bvartheta_{k+1}^{\sf fg} = \bvartheta_k^{\sf (1),fg} + \frac{h}{2}\bv_k^{\sf (2),fg}$ and $\bvartheta_{k+1}^{\sf sg} = \bvartheta_k^{\sf (1),sg} + \frac{h}{2}\bv_k^{\sf (2),sg}$, with $\bv_{k+1}^{\sf fg}=\bv_k^{\sf (2),fg}$ and $\bv_{k+1}^{\sf sg}=\bv_k^{\sf (2),sg}$.
		\end{itemize}
The assumption on the noisy gradient \textcolor{black}{(Assumption~\ref{asm:stoch_grad}, with the penalty gradient computed exactly)} gives
\begin{align*}
\E[\tilde \nabla U^\lambda (\bvartheta_k)|\bvartheta_k]=\nabla U^\lambda(\bvartheta_k),\qquad
\E[\|\tilde{\nabla} U^\lambda(\bvartheta_k)- \nabla U^\lambda(\bvartheta_k)\|^2|\bvartheta_k]\leqslant \textcolor{black}{\sigma_1^2L^2}\|\bvartheta_k\|^2.
\end{align*}
Writing $e_k:=\tilde\nabla U^\lambda(\bvartheta_k)-\nabla U^\lambda(\bvartheta_k)$, the cumulative difference contributed by each stage is given by the following expression.
		\begin{itemize}
			\item[(B)] $\Delta \bv_B=\bv_k^{\sf (1),sg}-\bv_k^{\sf (1),fg}=-h\,e_k$
                \item[(A)] $\Delta \bvartheta_A=\bvartheta_{k}^{\sf (1), sg}-\bvartheta_{k}^{\sf (1), fg}=\frac{h}{2}\Delta \bv_B=-\frac{h^2}{2}e_k$
                \item[(O)] $\Delta \bv_O=\bv_k^{\sf (2),sg}-\bv_k^{\sf (2),fg}=e^{-\gamma h}\Delta\bv_B=-he^{-\gamma h}e_k$
                \item[(A)] $\Delta \bvartheta_{A'}=\bvartheta_{k+1}^{\sf sg}-\bvartheta_{k+1}^{\sf fg}=\Delta\bvartheta_A+\frac{h}{2}\Delta\bv_O=-\frac{h^2}{2}(1+e^{-\gamma h})e_k$
		\end{itemize}
Thus $\Delta\bvartheta_{k+1}=-\frac{h^2}{2}(1+e^{-\gamma h})e_k$ and $\Delta\bv_{k+1}=-he^{-\gamma h}e_k$\textcolor{black}{, so that $T^{-1}\Delta Z_{k+1}=(\Delta\bvartheta_{k+1}-\frac h2\Delta\bv_{k+1},\Delta\bv_{k+1})=(-\frac{h^2}2e_k,-he^{-\gamma h}e_k)$. By Lemma~\ref{lem:norms}(i) and $h^2\leqslant\frac1{4M}$,
\[
\nsharp{\Delta Z_{k+1}}^2\leqslant\frac32\Big(\frac{h^4}4+\frac{h^2}M\Big)\|e_k\|^2\leqslant\frac32\cdot\frac{17}{16}\,\frac{h^2}{M}\|e_k\|^2\leqslant\frac85\,\frac{h^2}M\|e_k\|^2 ,
\]
and taking conditional expectations gives the claim.}
\end{proof}
 
We proceed to establish the bound for $ \wsharp(\nu_n^{\sf sg},\pi_h)$. Let $P_h$ denote the interior CBAOAB Markov kernel with full gradient $\nabla U^\lambda$ and friction $\gamma$, and let $\pi_h$ be its invariant law. Assume \eqref{eq:stepsize}.
 
\begin{enumerate}
  \item \textbf{Wasserstein contraction.} For all probability measures
  $\mu,\nu$ on $\mathbb R^{2p}$,
  \begin{equation}
    \label{eq:fg-contract}
    \wsharp(\mu P_h,\nu P_h)
    \;\leqslant\; (1-\rho(h))^{1/2}\,\wsharp(\mu,\nu),
    \qquad
    \rho(h) := \frac{m h^2}{4(1-e^{-\gamma h})}.
  \end{equation}
  More precisely, for synchronously coupled steps, $\nsharp{\Phi(z,\xi)-\Phi(z',\xi)}^2\leqslant(1-\rho(h))\nsharp{z-z'}^2$ for all $z,z',\xi$.
  \item \textbf{Bias to the target.} Let $\mu^{\rm int}$ be the law of $(\bvartheta,\bv+\frac h2\nabla U^\lambda(\bvartheta))$ for $(\bvartheta,\bv)\sim\pi^\lambda$, whose position marginal is $\nu^\lambda$. Then
  \begin{equation}
    \label{eq:fg-bias}
    \wsharp(\pi_h,\mu^{\rm int})
    \;\leqslant\;
    C_{\rm bias}\,(1-e^{-\gamma h}),
    \qquad
    C_{\rm bias}
    := \frac{16}{3}\,\frac{\sqrt{M p}}{m}.
  \end{equation}
  \item \textbf{Moments of the stationary chain.} If $Y_0\sim\pi_h$ and $Y_{k+1}=\Phi(Y_k,\xi_{k+1})$, then $Y_k\sim\pi_h$ for all $k$, and the position $\bvartheta^Y_k$ of $Y_k$ satisfies $\|\bvartheta_k^Y\|_{\Ltwo}\leqslant\mathsf m_\flat:=\sqrt{p/m}+8\mathsf B_h$.
\end{enumerate}
 
\noindent
\eqref{eq:fg-contract} is Proposition~\ref{prop:contraction}, \eqref{eq:fg-bias} is Proposition~\ref{prop:bias-cbaoab}, and item 3 follows from Lemma~\ref{lem:Ulam}(b) and Proposition~\ref{prop:bias-cbaoab}.
 
\medskip
 
\begin{proposition}
\label{prop:sg-baoab-helper}
Let Assumptions~\ref{asm:radius}--\ref{asm:smooth}, \ref{asm:stoch_grad} and \eqref{eq:stepsize} hold.
Define the constants $\rho(h)$, $C_{\rm bias}$ as in \eqref{eq:fg-contract}--\eqref{eq:fg-bias},
\begin{equation}
\label{eq:CV-DV-thm}
C_V := \frac85\,\frac{\sigma_1^2L^2}{M},\qquad \mathsf D_\flat:=\frac{13}8\sqrt{p/m}+\frac{16}3\mathsf B_h,
\end{equation}
and assume that the noise level satisfies the smallness condition
\begin{equation}
\label{eq:noise-small-thm}
C_V\,h^2 \leqslant \frac{\rho(h)^2}{54}.
\end{equation}
Define
\begin{equation}
\label{eq:lambda-sg-def}
\lambda_{\sf sg} := \frac{\rho(h)}{4},\qquad
C_{\rm mom} :=4\,\mathsf m_\flat^2+\frac92\,\mathsf D_\flat^2,\qquad
K_{\rm noise}^2 := C_VC_{\rm mom}.
\end{equation}
Then the SG chain satisfies the moment bound $\sup_{k\geqslant0}\E\|\bvartheta_k^{\sf sg}\|^2\leqslant C_{\rm mom}$, and, for every $n\geqslant 0$,
\begin{equation}
\label{eq:sg-baoab-W2-final-thm}
\wsharp(\nu_n^{\sf sg},\pi_h)
\leqslant \frac{\textcolor{black}{8}(1-e^{-\gamma h})}{\textcolor{black}{h}m}\,K_{\rm noise} + (1-\rho(h))^{n/2}\,\mathsf D_\flat ,
\end{equation}
where $\nu_n^{\sf sg}:=\nu Q_h^n$ and $\nu=\Law(Z_0)$ with $\bvartheta_0=\mathbf 0$, $\bv_0\sim\mathcal N(0,I_p)$.
\end{proposition}
 
\begin{proof}[Proof of Proposition~\ref{prop:sg-baoab-helper}]
We proceed in several steps.
 
\paragraph{Step 1: Use previous result.} Invoking Lemma~\ref{lem:sg-cbaoab}, we have
\begin{equation}
  \label{eq:DeltaZ-V}
  \mathbb E\bigl[\nsharp{\Delta Z_{k+1}}^2\mid\mathcal F_k\bigr]
  \leqslant C_V\,h^2\,\|\bvartheta_k^{\sf sg}\|^2 .
\end{equation}
 
\paragraph{Step 2: SG Lyapunov drift via the contraction of $P_h$.} Let $Y_0\sim\pi_h$ be optimally coupled with $Z_0^{\sf sg}$ for $\wsharp$, and let $Y_{k+1}=\Phi(Y_k,\xi_{k+1})$ be driven by the same Gaussian increments as the SG chain, so that $Y_k\sim\pi_h$. Set $V_k:=\nsharp{Z_k^{\sf sg}-Y_k}^2$. We have the decomposition
\[
  Z_{k+1}^{\sf sg}-Y_{k+1} = \big(\Phi(Z_k^{\sf sg},\xi_{k+1})-\Phi(Y_k,\xi_{k+1})\big) + \Delta Z_{k+1}.
\]
By Young's inequality, for any $\varepsilon>0$,
$$
  V_{k+1}
  \le (1+\varepsilon)\nsharp{\Phi(Z_k^{\sf sg},\xi_{k+1})-\Phi(Y_k,\xi_{k+1})}^2
     + \Bigl(1+\frac{1}{\varepsilon}\Bigr)\nsharp{\Delta Z_{k+1}}^2.
$$
Taking conditional expectations and using \eqref{eq:fg-contract} and \eqref{eq:DeltaZ-V}, together with $\|\bvartheta_k^{\sf sg}\|^2\leqslant2\|\bvartheta^Y_k\|^2+2\cdot\frac94V_k$ (Lemma~\ref{lem:norms}(iii)) and item 3, we get
\begin{equation}
  \label{eq:SG-drift-raw}
  \mathbb E\bigl[V_{k+1}\bigr]
  \leqslant A(\varepsilon,h)\,\mathbb E[V_k] + B(\varepsilon,h),
\end{equation}
where
$$
A(\varepsilon,h)
:= (1+\varepsilon)(1-\rho(h))
+ \Bigl(1+\frac{1}{\varepsilon}\Bigr)\frac92C_Vh^2,\qquad
B(\varepsilon,h)
:= \Bigl(1+\frac{1}{\varepsilon}\Bigr) 2C_Vh^2\,\mathsf m_\flat^2.
$$
We now choose
$ \varepsilon := \rho(h)/2$.
Then $(1+\varepsilon)(1-\rho(h))\leqslant1-\frac{\rho(h)}2$ and $1+\frac1\varepsilon\leqslant\frac3{\rho(h)}$, since $\rho(h)\leqslant1$.
Hence,
$$
A(\varepsilon,h)
\leqslant 1-\frac{\rho(h)}2
+ \frac{27\,C_Vh^2}{2\rho(h)}.
$$
Impose the small noise condition \eqref{eq:noise-small-thm}, i.e.
$
\frac{27C_Vh^2}{2\rho(h)}\leqslant\frac{\rho(h)}4,
$
which yields
$$
A(\varepsilon,h)\leqslant 1-\frac{\rho(h)}4=1-\lambda_{\sf sg}.
$$
Then,
\begin{equation}
 \label{eq:SG-drift-final}
\mathbb E\bigl[V_{k+1}\bigr]
\leqslant (1-\lambda_{\sf sg})\,\mathbb E[V_k]
+ \frac{6\,C_Vh^2\,\mathsf m_\flat^2}{\rho(h)}.
\end{equation}
 
\paragraph{Step 3: Uniform SG moment bound.}
Solving this linear recursion gives, for all $k\ge 0$,
$$
\mathbb E\bigl[V_k\bigr]
\leqslant (1-\lambda_{\sf sg})^k\mathbb E\bigl[V_0\bigr]
+ \frac{24\,C_Vh^2\,\mathsf m_\flat^2}{\rho(h)^2}
\leqslant \wsharp(\nu,\pi_h)^2+\frac49\mathsf m_\flat^2,
$$
using \eqref{eq:noise-small-thm}. Since $\wsharp(\nu,\pi_h)\leqslant\mathsf D_\flat$ (see the proof of Theorem~\ref{thm:cbaoab}), we obtain
$$
\sup_{k\geqslant 0}\mathbb E\bigl[\|\bvartheta_k^{\sf sg}\|^2\bigr]
\leqslant 2\mathsf m_\flat^2+\frac92\Big(\mathsf D_\flat^2+\frac49\mathsf m_\flat^2\Big)=\;C_{\rm mom}.
$$
 
\paragraph{Step 4: One-step kernel perturbation in $\wsharp$.}
Let $\mu$ be any probability measure on $\mathbb R^{2p}$ and
$Z_0\sim\mu$. From Step~1 we can couple
\[
  Y_1\sim P_h(Z_0,\cdot),\qquad
  Z_1^{\sf sg}\sim Q_h(Z_0,\cdot)
\]
using the same Gaussian noise so that
$\Delta Z_1:=Z_1^{\sf sg}-Y_1$ satisfies \eqref{eq:DeltaZ-V} with
$Z_k^{\sf sg}$ replaced by $Z_0$. Then
\[
  \wsharp^2(\mu Q_h,\mu P_h)
  \leqslant \mathbb E\bigl[\nsharp{Z_1^{\sf sg}-Y_1}^2\bigr]
  \leqslant C_Vh^2\,\E_\mu\|\bvartheta\|^2.
\]
If we take $\mu$ to be the law of $Z_k^{\sf sg}$, then by the uniform moment bound,
\begin{equation}
  \label{eq:kernel-perturb-W2}
  \wsharp(\mu Q_h,\mu P_h)
  \leqslant (C_VC_{\rm mom})^{1/2}h=K_{\rm noise}(h)\,h.
\end{equation}
 
\paragraph{Step 5: SG--FG chain distance.}
Let
$
  \nu_n^{\sf fg} := \nu P_h^n$ and $\nu_n^{\sf sg} := \nu Q_h^n$,
and define $E_n:=\wsharp(\nu_n^{\sf sg},\nu_n^{\sf fg})$. Then
\[
\begin{aligned}
  E_n
  &\le \wsharp(\nu Q_h^{n-1}Q_h,\nu Q_h^{n-1}P_h)
     + \wsharp(\nu Q_h^{n-1}P_h,\nu P_h^{n-1}P_h)\\
  &\leqslant \wsharp(\nu Q_h^{n-1}Q_h,\nu Q_h^{n-1}P_h)
     + (1-\rho(h))^{1/2}E_{n-1},
\end{aligned}
\]
where we used the contraction \eqref{eq:fg-contract} for the second term.
Applying \eqref{eq:kernel-perturb-W2} with $\mu=\nu Q_h^{n-1}$, we obtain
$$
E_n \leqslant K_{\rm noise}(h)h + (1-\rho(h))^{1/2}E_{n-1},\qquad E_0=0.
$$
Solving this recursion yields
$$
E_n \leqslant K_{\rm noise}(h)h \sum_{j=0}^{n-1}(1-\rho(h))^{j/2}
\le \frac{2K_{\rm noise}(h)h}{\rho(h)},
$$
using $1-\sqrt{1-\rho}\geqslant\rho/2$.
With $\rho(h)$ as in \eqref{eq:fg-contract}, this gives
$$
\wsharp(\nu_n^{\sf sg},\nu_n^{\sf fg})
\leqslant \frac{\textcolor{black}{8}(1-e^{-\gamma h})}{\textcolor{black}{h}m}\,K_{\rm noise}(h).
$$

\paragraph{Step 6: Final bound to $\pi_h$.}
Finally, decompose
$$
\wsharp(\nu_n^{\sf sg}, \pi_h)
\leqslant \wsharp(\nu_n^{\sf sg},\nu_n^{\sf fg}) + \wsharp(\nu_n^{\sf fg},\pi_h).
$$
We have just shown
$
\wsharp(\nu_n^{\sf sg},\nu_n^{\sf fg})
\leqslant \frac{8(1-e^{-\gamma h})}{hm}\,K_{\rm noise}(h).
$
From \eqref{eq:fg-contract} and the fact that $\pi_h$ is invariant for $P_h$,
$$
\wsharp(\nu_n^{\sf fg},\pi_h)
\leqslant (1-\rho(h))^{n/2}\wsharp(\nu,\pi_h),
$$
and, by the triangle inequality and \eqref{eq:fg-bias},
$
\wsharp(\nu,\pi_h)
\leqslant \wsharp(\nu,\textcolor{black}{\mu^{\rm int}}) + \wsharp(\mu^{\rm int},\pi_h)
\leqslant \frac{13}8\sqrt{p/m}+ C_{\rm bias}(1-e^{-\gamma h})=\mathsf D_\flat.
$
Combining these bounds yields \eqref{eq:sg-baoab-W2-final-thm}, which completes the proof.
\end{proof}
 
\begin{proof}[Proof of Theorem~\ref{thm:CBAOAB_SG}]
The proof of Theorem~\ref{thm:CBAOAB_SG} follows directly from Proposition~\ref{prop:sg-baoab-helper}, and using the triangle inequality and Lemma~\ref{lem:norms}: since the position of $Z_n^{\sf sg}$ is $\bvartheta_n$,
$\wstwo(\nu_n,\nu)\leqslant\frac32\wsharp(\nu_n^{\sf sg},\pi_h)+\wstwo(\pi_h^{\mathsf x},\nu^\lambda)+\wstwo(\nu^\lambda,\nu)$, where $\wstwo(\pi_h^{\mathsf x},\nu^\lambda)\leqslant8\mathsf B_h$ by Proposition~\ref{prop:bias-cbaoab}, and the last term is bounded by Proposition~\ref{prop:w2_tight}. Moreover $\frac32\mathsf D_\flat\leqslant\frac52\sqrt{p/m}+8\mathsf B_h$, $\frac32\cdot8\sqrt{8/5}\leqslant16$, $C_{\rm mom}^{1/2}\leqslant2\mathsf m_\flat+\frac3{\sqrt2}\mathsf D_\flat\leqslant\mathsf S_\flat$, and \eqref{eq:noise-small-thm} is equivalent to \eqref{eq:small-noise-baoab} up to rounding of the constant.
\end{proof}
 
\subsubsection*{Proofs of convergence for SG-CABOBA}

For SG-CABOBA a fresh stochastic gradient is drawn at every step, so $Z_k=(\bvartheta_k,\bv_k)$ is a Markov chain. Let $P_h$ and $Q_h$ denote the CABOBA and SG-CABOBA kernels, $\pi_h:=\pi_h^{\rm ABOBA}$, $\nu:=\delta_{\mathbf 0}\otimes\mathcal N(0,I_p)$, and $\tilde\bvartheta_k:=\bvartheta_k+\frac h2\bv_k$. The one-step discrepancy $\Delta Z_{k+1}$ is defined as for SG-CBAOAB.

\begin{lemma}
\label{lem:sg-caboba}
When \eqref{eq:stepsize} holds,
$\E[\nsharp{\Delta Z_{k+1}}^2\mid\mathcal F_k]\leqslant C_V\,h^2\,\|\tilde\bvartheta_k^{\sf sg}\|^2$ with $C_V:=\sigma_1^2L^2/M$.
\end{lemma}
\begin{proof}
Let $e_k$ be the error of the stochastic gradient at $\tilde\bvartheta_k$, which is used in both half kicks. The stage-by-stage differences are
\[
\text{(A) } \Delta\tilde\bvartheta=\mathbf 0,\quad
\text{(B) } \Delta\bv_B=-\tfrac h2e_k,\quad
\text{(O) } \Delta\bv_O=-\tfrac h2\etah e_k,\quad
\text{(B) } \Delta\bv_{k+1}=-\tfrac h2(1+\etah)e_k,\quad
\text{(A) } \Delta\bvartheta_{k+1}=\tfrac h2\Delta\bv_{k+1}.
\]
Hence $T^{-1}\Delta Z_{k+1}=(\mathbf 0,\Delta\bv_{k+1})$ and $\nsharp{\Delta Z_{k+1}}^2=a\frac{h^2}4(1+\etah)^2\|e_k\|^2\leqslant\frac{h^2}M\|e_k\|^2$. Assumption~\ref{asm:stoch_grad}(ii) gives the claim.
\end{proof}

\begin{proposition}
\label{prop:sg-aboba-helper}
Let Assumptions~\ref{asm:radius}--\ref{asm:smooth}, \ref{asm:stoch_grad} and \eqref{eq:stepsize} hold, and assume $C_Vh^2\leqslant\rho(h)^2/60$. Define
\[
\mathsf m_\sharp:=\tfrac{21}{20}\sqrt{p/m}+8\mathsf B_h,\qquad \mathsf D_\sharp:=\sqrt{p/m}+\tfrac{24}5\mathsf B_h,\qquad C_{\rm mom}:=4\mathsf m_\sharp^2+5\mathsf D_\sharp^2,\qquad K_{\rm noise}^2:=C_VC_{\rm mom}.
\]
Then $\sup_{k\geqslant0}\E\|\tilde\bvartheta_k^{\sf sg}\|^2\leqslant C_{\rm mom}$, and for every $n\geqslant0$,
\[
\wsharp(\nu Q_h^n,\pi_h)\leqslant\frac{8(1-e^{-\gamma h})}{hm}\,K_{\rm noise}+(1-\rho(h))^{n/2}\,\mathsf D_\sharp .
\]
\end{proposition}
\begin{proof}
We follow Steps 1--6 of the proof of Proposition~\ref{prop:sg-baoab-helper}, with Lemma~\ref{lem:sg-caboba} in place of Lemma~\ref{lem:sg-cbaoab}, $\tilde\bvartheta$ in place of $\bvartheta$, and the following facts for the CABOBA kernel: the contraction \eqref{eq:fg-contract} (Proposition~\ref{prop:contraction}), $\wsharp(\pi_h,\pi^\lambda)\leqslant\frac{24}5\mathsf B_h$ (Proposition~\ref{prop:bias-cbaoab}), and $\|\tilde\bvartheta\|_{\Ltwo(\pi_h)}\leqslant\mathsf m_\sharp$, which follows from $\|\tilde\bvartheta\|_{\Ltwo(\pi^\lambda)}^2\leqslant\frac pm+s^2p\leqslant\frac{17}{16}\frac pm$ (Lemma~\ref{lem:Ulam}(ii) and \eqref{eq:s-small}) and Lemma~\ref{lem:norms}(c). The only changes are the constants.
In Step 2, Lemma~\ref{lem:norms}(iii) gives $\|\tilde\bvartheta_k^{\sf sg}\|^2\leqslant2\|\tilde\bvartheta_k^Y\|^2+5V_k$, so that
\[
A(\varepsilon,h)\leqslant1-\frac{\rho(h)}2+\frac{15\,C_Vh^2}{\rho(h)}\leqslant1-\frac{\rho(h)}4,\qquad B(\varepsilon,h)\leqslant\frac{6\,C_Vh^2\,\mathsf m_\sharp^2}{\rho(h)} .
\]
In Step 3, this gives $\E V_k\leqslant\mathsf D_\sharp^2+\frac25\mathsf m_\sharp^2$, using $\wsharp(\nu,\pi_h)\leqslant\mathsf D_\sharp$ (proof of Theorem~\ref{thm:caboba}), and hence the moment bound. Steps 4--6 are unchanged.
\end{proof}

\begin{proof}[Proof of Theorem~\ref{thm:CABOBA_SG}]
\textcolor{black}{As in the proof of Theorem~\ref{thm:CBAOAB_SG}, using Proposition~\ref{prop:sg-aboba-helper}, $\wstwo(\pi_h^{\mathsf x},\nu^\lambda)\leqslant\frac{36}5\mathsf B_h$, $\frac32\mathsf D_\sharp=\frac32\sqrt{p/m}+\frac{36}5\mathsf B_h$, $\frac32\cdot8=12$ and $C_{\rm mom}^{1/2}\leqslant2\mathsf m_\sharp+\sqrt5\,\mathsf D_\sharp\leqslant\mathsf S_\sharp$. Finally, $C_Vh^2\leqslant\rho(h)^2/60$ is implied by \eqref{eq:small-noise-aboba}.}
\end{proof}

 \begin{proof}[Proof of Corollaries~\ref{cor:cbaoab}, \ref{cor:caboba}, \ref{cor:sg-cbaoab} and \ref{cor:sg-caboba}]
\emph{Preliminaries.} Let $\gamma=4\sqrt M$ and $h\leqslant1/(4\sqrt M)$, so that $t:=\gamma h\in(0,1]$. Since $1-e^{-t}>t-\frac{t^2}2\geqslant\frac t2=2h\sqrt M$, the step-size condition \eqref{eq:stepsize} holds, and
\[
\mathsf B_h\leqslant\frac{4hM\sqrt p}{m},\qquad c(h)\geqslant\frac{mh}{16\sqrt M},\qquad M=M^\lambda\leqslant\frac{L+C_{\rm proj}}{\lambda^2}\quad(\lambda\leqslant1).
\]
Hence $\mathsf B_h=\bigo(h\lambda^{-2})$, and the transient term is $\bigo(\varepsilon)$ once $n\,c(h)\gtrsim\log(1/\varepsilon)$, that is, for $n=\widetilde\Omega\big((\lambda h)^{-1}\big)$.

\emph{Bound in $\wstwo$.} Take $\lambda=\Theta(\varepsilon^{2p/(p+2)})$, so that $C(p,2)\lambda^{1/2+1/p}=\bigo(\varepsilon)$, and $h=\Theta(\varepsilon\lambda^2)=\Theta(\varepsilon^{(5p+2)/(p+2)})$, so that $\mathsf B_h=\bigo(\varepsilon)$. For small $\varepsilon$ this $h$ satisfies $h\leqslant1/(4\sqrt M)$, since $1/\sqrt M\gtrsim\lambda\gg\varepsilon\lambda^2$. The number of steps is then $n=\widetilde\Omega(\varepsilon^{-(7p+2)/(p+2)})$.

\emph{Bound in $\wsone$.} Take $\lambda=\Theta(\varepsilon)$ and $h=\Theta(\varepsilon^3)$. Then $C(p,1)\lambda=\bigo(\varepsilon)$, $\mathsf B_h=\bigo(\varepsilon)$, and $n=\widetilde\Omega(\varepsilon^{-4})$.

\emph{Stochastic gradients.} With the same choices of $\lambda$, $h$ and $n$, we have $\mathsf S_\flat,\mathsf S_\sharp=\bigo(1)$ and, since $1-e^{-\gamma h}\leqslant4h\sqrt M$,
\[
\frac{(1-e^{-\gamma h})\,\sigma_1L}{hm\sqrt M}\leqslant\frac{4\sigma_1L}{m}=\bigo(b^{-1/2}),
\]
which is $\bigo(\varepsilon)$ for $b=\Theta(\varepsilon^{-2})$. The small-noise conditions \eqref{eq:small-noise-baoab}--\eqref{eq:small-noise-aboba} then hold for small $\varepsilon$, since their right-hand sides are at least $m/152$. The total number of stochastic gradient evaluations is $nb$, that is, $\widetilde\bigo(\varepsilon^{-(9p+6)/(p+2)})$ for $\wstwo$ and $\widetilde\bigo(\varepsilon^{-6})$ for $\wsone$.
\end{proof}

\subsection{Proofs of convergence for SG-CKLMC}
\label{app:sgcklmc}
This appendix section is devoted to the proof of Theorem~\ref{thm:CKLMC_SG}. The proof technique is similar to that in \cite{Mert24}, but here we provide an explicit expression for the Wasserstein distance and combine it with our refined analysis of the distance between the surrogate distribution and the target distribution.
To prove Theorem~\ref{thm:CKLMC_SG}, we first establish the following lemma, which provides an upper bound on the second moment of the $n$-th iterate $\vartheta_n$ generated by the SG-CKLMC algorithm.
For notational simplicity, throughout this section we omit the superscript $\mathsf{SG\text{-}CKLMC}$ whenever no ambiguity arises.
 
\begin{lemma}
\label{lem:sg-cklmc_mom}
Assume that the step size $h$ satisfies
\[
h\leqslant \min\left\{\frac{\gamma
\tau }{2K_1},\frac{2}{\gamma \tau},\frac{1}{10\gamma}\right\}\,,
\]
where
\begin{align*}
\tau
&= \frac{1}{2}\min\left\{ \frac{1}{4}, \frac{m}{M^\lambda + \gamma^2/2} \right\}, \quad
K_1
= \max\left\{
\frac{16\big((M^\lambda)^2 + 2\gamma (M^\lambda)^2 + \sigma_1^2 L^2\big)}{(1-2\tau)\gamma^2},
\frac{4M^\lambda + 2\gamma^2(1-\tau) + 8\gamma}{1-2\tau}
\right\}.
\end{align*}
Then the SG–CKLMC iterates $\{\bvartheta_n\}_{n\geqslant 0}$ satisfy the uniform second moment bound
\begin{align*}
\|\bvartheta_n\|^2_{\Ltwo}\leqslant \frac{8}{(1-2\tau)\gamma^2}C_{\mathcal{V}},\qquad n\geqslant 0\,,
\end{align*}
where
\[
C_{\mathcal{V}}=\int_{\mathbb{R}^{2p}}\mathcal{V}(\bvartheta,\bv)\, \mu_0(d\bvartheta,d\bv)+ \frac{4}{\tau}\left(p+\frac{mf(0)}{2M^\lambda +\gamma^2}\right)\,,
\]
and the Lyapunov function $\mathcal{V}$ is defined by
\[
\mathcal{V}(\bvartheta,\bv)=U^\lambda(\bvartheta)+\frac{\gamma^2}{4}(\|\bvartheta+\gamma^{-1}\bv\|^2+\|\gamma^{-1}\bv\|^2-\tau\|\bvartheta\|^2)\,.
\]
\end{lemma}
 
\begin{proof}
The proof follows the same steps as Proposition 2.22 in~\cite{Mert24}, relying on Lemma EC.5 in~\cite{gao2022global}.
\end{proof}
 
We are now ready to prove Theorem~\ref{thm:CKLMC_SG}.
\begin{proof}[Proof of Theorem~\ref{thm:CKLMC_SG}]
Let $\bvartheta_n$, $\bv_n$ be the iterates of the CKLMC algorithm with a stochastic gradient.
Let $(\bL_t,\bV_t)$ be the
kinetic Langevin diffusion, coupled with
$(\bvartheta_n, \bv_n)$ through the same Brownian
motion $(\bW_t; t\geqslant 0)$ and starting from
a random point $(\bL_0, \bV_0)\sim\pi^\lambda$ such that
$\bV_0=\bv_0$. This means that
\begin{align*}
    \bv_{n+1} & = e^{-\gamma h}\bv_n   -  \frac{1-e^{-\gamma h}}{\gamma}
    \widetilde\nabla U^\lambda(\bvartheta_n) + \sqrt{2\gamma } \xi_{n+1}\\
    \bvartheta_{n+1} & = \bvartheta_n + \frac{1-e^{-\gamma h}}{\gamma} \bv_n
    - \frac{e^{-\gamma h}-1+\gamma h}{\gamma^2}
    \widetilde\nabla U^\lambda(\bvartheta_n)
    + \sqrt{2\gamma}\xi'_{n+1}\,.
\end{align*}
To proceed with the proof, we introduce an auxiliary process.
\begin{align*}
    \tilde\bv_{n+1} & = e^{-\gamma h}\bv_n   -  \frac{1-e^{-\gamma h}}{\gamma}
    \nabla U^\lambda(\bvartheta_n) + \sqrt{2\gamma } \xi_{n+1}\\
    \tilde\bvartheta_{n+1} & = \bvartheta_n + \frac{1-e^{-\gamma h}}{\gamma} \bv_n
    -  \frac{e^{-\gamma h}-1+\gamma h}{\gamma^2}
    \nabla U^\lambda(\bvartheta_n) + \sqrt{2\gamma} \xi'_{n+1}\,.
\end{align*}
We set $\tilde \bv_0=\bv_0$ and $ \tilde\bvartheta_0=\bvartheta_0.$
Our goal will be to bound the term $x_n$ defined by
\begin{align}\label{eq:xn}
    x_n = \bigg\| \bfC^{-1}
    \begin{bmatrix}
        \bv_n- \bV_{nh}\\
        \bvartheta_n-\bL_{nh}
    \end{bmatrix}
    \bigg\|_{\Ltwo}
    \quad \text{with}\quad
    \bfC = \frac{1}{\gamma}\begin{bmatrix}
    \mathbf 0_{p\times p} & -\gamma \mathbf I_{p}\\
    \mathbf I_p & \mathbf I_p
    \end{bmatrix}.
\end{align}
Applying the triangle inequality, we have
\begin{align}
\label{eq:cklmc}
x_{n+1} \leqslant
\bigg\| \bfC^{-1}
    \begin{bmatrix}
        \bv_{n+1}- \tilde\bv_{n+1}\\
        \bvartheta_{n+1}-\tilde\bvartheta_{n+1}
    \end{bmatrix}
    \bigg\|_{\Ltwo}+
    \bigg\| \bfC^{-1}
    \begin{bmatrix}
        \tilde\bv_{n+1} - \bV_{(n+1)h} \\
        \tilde\bvartheta_{n+1}- \bL_{(n+1)h}
    \end{bmatrix}
    \bigg\|_{\Ltwo}
\end{align}
By Lemma~\ref{lem:sg-cklmc_mom} and $\gamma h<0.1$, we have
\begin{align*}
\bigg\| \bfC^{-1}
    \begin{bmatrix}
        \bv_{n+1}- \tilde\bv_{n+1}\\
        \bvartheta_{n+1}-\tilde\bvartheta_{n+1}
    \end{bmatrix}
    \bigg\|_{\Ltwo}
&\leqslant 2\|\bv_{n+1}- \tilde\bv_{n+1}\|_{\Ltwo}
+ \gamma \| \bvartheta_{n+1}-\tilde\bvartheta_{n+1}\|_{\Ltwo}
\leqslant 2.1h\sigma_1 L \|\bvartheta_n\|_{\Ltwo}
\\&\leqslant  2.1h\sigma_1 L\sqrt{\frac{8}{(1-2\tau)\gamma^2}C_{\mathcal{V}}}\,.
\end{align*}
By Theorem 2 in \cite{dalalyan2020sampling}, we obtain
\begin{align*}
 \bigg\| \bfC^{-1}
    \begin{bmatrix}
        \tilde\bv_{n+1} - \bV_{(n+1)h} \\
        \tilde\bvartheta_{n+1}- \bL_{(n+1)h}
    \end{bmatrix}
    \bigg\|_{\Ltwo}
    \leqslant \bigg(1-\frac{0.75mh}{\gamma}\bigg)  x_n
    + 0.75M^\lambda h^2\sqrt{p}\,,
\end{align*}
provided that $h\leqslant \frac{m}{4\gamma M^\lambda}$ and $\gamma\geqslant \sqrt{M^\lambda+m}.$
Combining these two displays with inequality~\eqref{eq:cklmc} gives
\begin{align*}
x_{n+1}\leqslant \bigg(1-\frac{0.75mh}{\gamma}\bigg) x_n +
0.75M^\lambda h^2\sqrt{p}
+ 2.1h\sigma_1 L \sqrt{\frac{8C_{\mathcal{V}}}{(1-2\tau)\gamma^2}}\,.
\end{align*}
It then follows by induction that
\begin{align*}
x_n\leqslant  \bigg(1-\frac{0.75mh}{\gamma}\bigg)^n x_0 +
\frac{ \gamma M^\lambda h\sqrt{p}}{m}
+\frac{8\sigma_1 L}{m} \sqrt{\frac{C_{\mathcal{V}}}{1-2\tau}} \,.
\end{align*}
Note that $x_0=\gamma \wstwo(\nu_0,\nu^\lambda)$ and $\wstwo(\nu_n,\nu^\lambda)
\leqslant \| \bvartheta_n-\bL_{nh}\|_{\Ltwo}
\leqslant \gamma^{-1}\sqrt{2}x_n.$
We then have
\begin{align*}
\wstwo(\nu_n,\nu^\lambda)\leqslant
 \bigg(1-\frac{0.75mh}{\gamma}\bigg)^n  \wstwo(\nu_0,\nu^\lambda)
 + \frac{ \sqrt{2}M^\lambda h\sqrt{p}}{m}
 + \frac{8\sqrt{2}\sigma_1 L}{m \gamma} \sqrt{\frac{C_{\mathcal{V}}}{1-2\tau}} \,.
\end{align*}
When $p>2$, applying the triangle inequality and Proposition~\ref{prop:w2_tight}, together with $\wstwo(\nu_0,\nu^\lambda)\leqslant\wstwo(\nu_0,\nu)+\wstwo(\nu,\nu^\lambda)$, gives the claim. The bound for the Wasserstein-1 distance follows similarly, using $\wsone\leqslant\wstwo$ for the sampling error.
\end{proof}

\newpage

\section{Additional Experiments}

\subsection{Triangular Constraint}
\label{sec:num-tri}

Our second example uses the same target density $f(\btheta)=\|\btheta\|^2/2$ on $\R^2$, restricted to the triangle
\begin{equation}
\label{eqn:convex_set_t}
\mathcal{K} = \Big\{ (x_1,x_2) \in \mathbb{R}^2:  x_1,x_2 \geq -0.3, \quad x_1+x_2 \leq 0.3  \Big\}=: \triangle ,
\end{equation}
with vertices $(-0.3,-0.3)$, $(0.6,-0.3)$ and $(-0.3,0.6)$. Unlike the disc, this set is not centered at the origin and its boundary has corners. We use the Gauge projection (Section~\ref{sec:back}), with gauge function $g_\K(\btheta)=\max_i\langle a_i,\btheta\rangle$ with $a_1=(0,-\frac1{0.3})$, $a_2=(-\frac1{0.3},0)$, $a_3=(\frac1{0.3},\frac1{0.3})$, and the penalty $d_\K(\btheta)=\|\btheta-\btheta/g_\K(\btheta)\|^2$ for $g_\K(\btheta)>1$. All other settings are as in Section~\ref{sec:num-circle}: $\lambda=0.03$, $\gamma=2$, $T=8$, $10^4$ chains, 3 repetitions, and stochastic gradients with batch size $64$ from the same finite-sum representation of $f$. We report step sizes relative to the nominal value $M^\lambda=1+\lambda^{-2}$.

Since the laws are not rotation invariant, we measure errors with the sliced Wasserstein distance $SW_2(\mu,\nu)^2=\frac1{64}\sum_{j=1}^{64}\wstwo^2(\langle u_j,\cdot\rangle_\#\mu,\langle u_j,\cdot\rangle_\#\nu)$ over $64$ equispaced directions $u_j$, computed against the exact density of $\nu^\lambda$ on a fine grid, after subtraction of the Monte Carlo floor. For this example $\nu^\lambda$ puts $17\%$ of its mass outside $\K$, and $SW_2(\nu^\lambda,\nu)\approx0.016$.

\begin{figure}[h!]
  \centering
  \includegraphics[width=\linewidth]{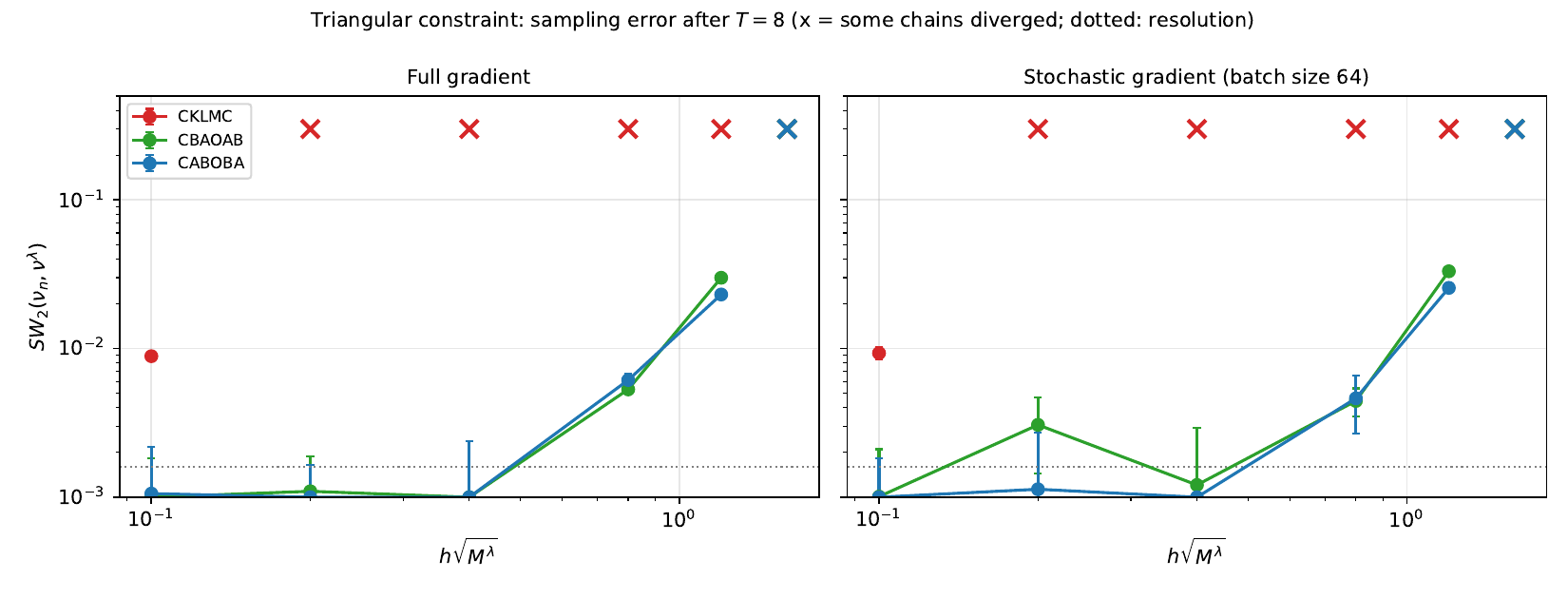}
  \caption{Triangular constraint: sampling error $SW_2(\nu_n,\nu^\lambda)$ after time $T=8$ as a function of the step size, for full gradients (left) and stochastic gradients with batch size $64$ (right). Values below the Monte Carlo resolution are shown at $10^{-3}$. The crosses indicate step sizes at which some chains diverged.}
\label{fig:tri-error}
\end{figure}

\begin{figure}[h!]
  \centering
  \includegraphics[width=\linewidth]{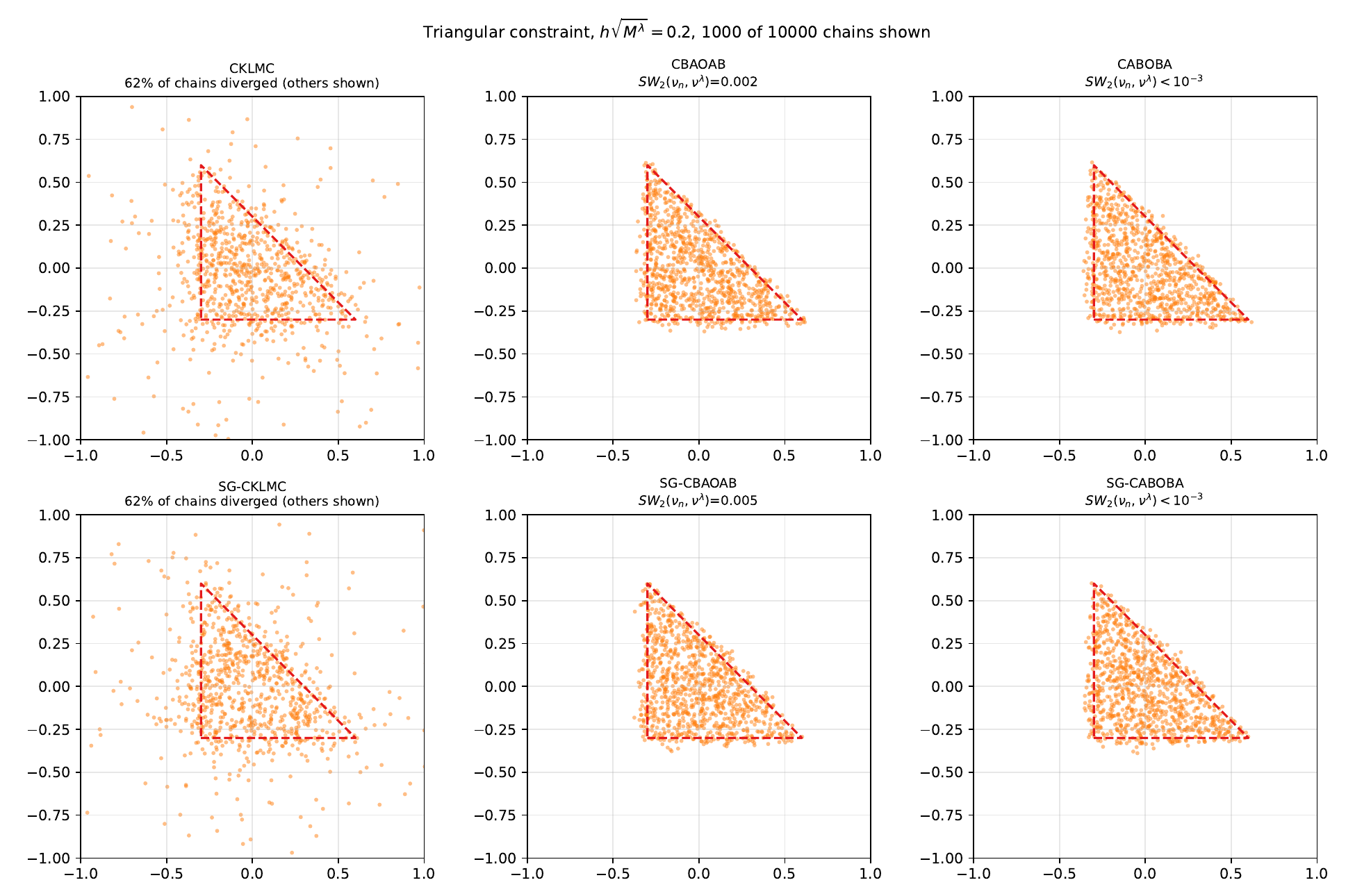}
  \caption{Triangular constraint: last iterates of $1000$ chains for $h\sqrt{M^\lambda}=0.2$, with full gradients (top) and stochastic gradients (bottom). For CKLMC and SG-CKLMC, $62\%$ of the chains diverge and the remaining ones are shown.}
\label{fig:tri-scatter}
\end{figure}

\begin{table}[h!]
\begin{center}
\begin{tabular}{|l|c|c|c|}
\hline
\textbf{Algorithm} & $h\sqrt{M^\lambda}=0.1$ & $h\sqrt{M^\lambda}=0.4$ & $h\sqrt{M^\lambda}=1.2$ \\
\hline
CKLMC & 0.0089 & unstable & unstable \\
CBAOAB & $<10^{-3}$ & $<10^{-3}$ & 0.0299 \\
CABOBA & 0.0011 & $<10^{-3}$ & 0.0231 \\
SG-CKLMC & 0.0093 & unstable & unstable \\
SG-CBAOAB & 0.0010 & 0.0012 & 0.0331 \\
SG-CABOBA & $<10^{-3}$ & $<10^{-3}$ & 0.0256 \\
\hline
\end{tabular}

\end{center}
\caption{Triangular constraint: $SW_2(\nu_n,\nu^\lambda)$ after time $T=8$ for three step sizes (mean over 3 repetitions, Monte Carlo resolution $\approx0.003$).}
\label{table:tri}
\end{table}

The results, shown in Figures~\ref{fig:tri-error}--\ref{fig:tri-scatter} and Table~\ref{table:tri}, are consistent with those for the disc. CKLMC and SG-CKLMC are only stable for $h\sqrt{M^\lambda}=0.1$, where their error is about $0.009$, and $62\%$ of their chains diverge for $h\sqrt{M^\lambda}=0.2$. The splitting schemes are stable up to $h\sqrt{M^\lambda}=1.2$. Their error is at the Monte Carlo resolution up to $h\sqrt{M^\lambda}=0.4$ and grows for larger step sizes. They become unstable for $h\sqrt{M^\lambda}=1.6$, earlier than for the disc, which reflects the fact that the Gauge penalty of this triangle is stiffer than the nominal value $M^\lambda=1+\lambda^{-2}$. %The stochastic gradient versions again behave like the full-gradient versions.

\newpage

\section{Algorithms}
\label{sec:app_alg}
 
In all algorithms, $\widetilde\nabla U^\lambda(\btheta,\omega)=\mathcal G(\btheta,\omega)+\nabla\ell^\lambda_\K(\btheta)$ denotes the stochastic gradient of the surrogate potential, and $\etah=e^{-\gamma h}$. The full-gradient versions are obtained by replacing $\widetilde\nabla U^\lambda$ with $\nabla U^\lambda$.
 
\begin{algorithm}[h!]
    \footnotesize
	\begin{itemize}
		\item Initialize $\left(\bvartheta_{0},\bv_{0}\right) \in \mathbb{R}^{2p}$, stepsize $h > 0$ and friction parameter $\gamma > 0$.
		\item for $k = 1,2,...,K$ do
		\begin{itemize}
                \item[(A)] $\tilde\bvartheta \leftarrow \bvartheta_{k-1} + \frac{h}{2}\bv_{k-1}$
                \item[] Sample $\omega_{k} \sim \rho$ and set $G_k \leftarrow \widetilde\nabla U^\lambda(\tilde\bvartheta,\omega_{k})$
			\item[(B)] $\bv \leftarrow \bv_{k-1} - \frac{h}{2}G_{k}$
                \item[] Sample $\xi_{k} \sim \mathcal{N}(0_{p},I_{p})$
                \item[(O)] $\bv \leftarrow \etah \bv + \sqrt{1-\etah^{2}}\,\xi_{k}$
                \item[(B)] $\bv_{k} \leftarrow \bv - \frac{h}{2}G_{k}$
                \item[(A)] $\bvartheta_{k} \leftarrow \tilde\bvartheta + \frac{h}{2}\bv_k$
		\end{itemize}
            \item Output: Samples $(\bvartheta_{k})^{K}_{k=0}$.
	\end{itemize}
	\caption{Stochastic Gradient Constrained ABOBA (SG-CABOBA)}
	\label{alg:SG-CABOBA}
\end{algorithm}
 
\begin{algorithm}[h!]
    \footnotesize
	\begin{itemize}
		\item Initialize $\left(\bvartheta_{0},\bv_{0}\right) \in \mathbb{R}^{2p}$, stepsize $h > 0$ and friction parameter $\gamma > 0$.
            \item Sample $\omega_{1} \sim \rho$ and set $G_{0} \leftarrow \widetilde\nabla U^\lambda(\bvartheta_{0},\omega_{1})$
		\item for $k = 1,2,...,K$ do
		\begin{itemize}
			\item[(B)] $\bv \leftarrow \bv_{k-1} - \frac{h}{2}G_{k-1}$
                \item[(A)] $\bvartheta \leftarrow \bvartheta_{k-1} + \frac{h}{2}\bv$
                \item[] Sample $\xi_{k} \sim \mathcal{N}(0_{p},I_{p})$
                \item[(O)] $\bv \leftarrow \etah \bv + \sqrt{1-\etah^{2}}\,\xi_{k}$
                \item[(A)] $\bvartheta_{k} \leftarrow \bvartheta + \frac{h}{2}\bv$
                \item[] Sample $\omega_{k+1} \sim \rho$ and set $G_{k} \leftarrow \widetilde\nabla U^\lambda(\bvartheta_{k},\omega_{k+1})$
                \item[(B)] $\bv_{k} \leftarrow \bv - \frac{h}{2}G_{k}$
		\end{itemize}
            \item Output: Samples $(\bvartheta_{k})^{K}_{k=0}$.
	\end{itemize}
	\caption{Stochastic Gradient Constrained BAOAB (SG-CBAOAB)}
	\label{alg:SG-CBAOAB}
\end{algorithm}
 
\begin{algorithm}[h!]
    \footnotesize
	\begin{itemize}
		\item Initialize $\left(\bvartheta_{0},\bv_{0}\right) \in \mathbb{R}^{2p}$, stepsize $h > 0$ and friction parameter $\gamma > 0$.
		\item for $k = 1,2,...,K$ do
		\begin{itemize}
                \item[] Sample $\omega_{k} \sim \rho$
                \item[] Sample $\xi_{k} \sim \mathcal{N}(0_{p},I_{p})$
                \item[] $\bvartheta_{k} \leftarrow \bvartheta_{k-1} + h \bv_{k-1}$
                \item[] $\bv_{k} \leftarrow \bv_{k-1} - h\, \widetilde\nabla U^\lambda(\bvartheta_{k-1},\omega_{k}) - h \gamma \bv_{k-1} + \sqrt{2\gamma h}\,\xi_{k}$
		\end{itemize}
            \item Output: Samples $(\bvartheta_{k})^{K}_{k=0}$.
	\end{itemize}
	\caption{Stochastic Gradient Euler-Maruyama (SG-CKLMC)}
	\label{alg:SG-CEM}
\end{algorithm}
 
\end{document}